\documentclass[twocolumn,floatfix,showpacs,amsmath,amssymb,pra,superscriptaddress]{revtex4}

\usepackage{dcolumn} 
\usepackage{amsmath,bbm}
\usepackage{amsfonts,amssymb}
\usepackage{graphicx}
\usepackage{epsfig}
\usepackage{times}
\usepackage{verbatim}
\usepackage{color}

\newcommand{\be}{\begin{equation}}
\newcommand{\ee}{\end{equation}}
\newcommand{\bq}{\begin{eqnarray}}
\newcommand{\eq}{\end{eqnarray}}
\newcommand{\ba}{\begin{align}}
\newcommand{\ea}{\end{align}}

\newcommand{\1}{\mathbbm{1}}

\newcommand{\ket}[1]{ | \, #1 \rangle}
\newcommand{\bra}[1]{ \langle #1 \,  |}

\newcommand{\tr}[1]{{\rm tr}\left[{#1}\right]}

\newcommand{\bZ}{\mathbbm{Z}}

\newcommand{\cE}{\mathcal{E}}

\newcommand{\cL}{\mathcal{L}}

\newcommand{\cM}{\mathcal M}

\newcommand{\cO}{\mathcal O}
\newcommand{\cV}{\mathcal V}

\newtheorem{theorem}{Theorem}
\newtheorem{lemma}[theorem]{Lemma}

\newtheorem{definition}[theorem]{Definition}

\def\qed{\leavevmode\unskip\penalty9999 \hbox{}\nobreak\hfill
     \quad\hbox{\leavevmode  \hbox to.77778em{%
               \hfil\vrule   \vbox to.675em%
               {\hrule width.6em\vfil\hrule}\vrule\hfil}}
     \par\vskip3pt}
    {\hspace*{\fill}$\Box$\vspace{1.5ex}\par}

\newcommand{\Sp}{\,\,\,\,\,\,}

\begin{document}

\title{Self correction requires Energy Barrier for Abelian quantum doubles}

\author{Anna K\'om\'ar}
\affiliation{Institute for Quantum Information and Matter and Walter Burke Institute for Theoretical Physics, California Institute of Technology, Pasadena, California 91125, USA}
\author{Olivier Landon-Cardinal} 
\affiliation{Institute for Quantum Information and Matter and Walter Burke Institute for Theoretical Physics, California Institute of Technology, Pasadena, California 91125, USA}
\affiliation{Department of Physics, McGill University, Montreal, Canada H3A 2T8}
\author{Kristan Temme}
\affiliation{Institute for Quantum Information and Matter and Walter Burke Institute for Theoretical Physics, California Institute of Technology, Pasadena, California 91125, USA}
\affiliation{IBM T.J. Watson Research Center, Yorktown Heights, NY 10598, USA}
\date{\today}

\begin{abstract} 
We rigorously establish an Arrhenius law for the mixing time of quantum doubles based on any Abelian group $\mathbb{Z}_d$. We have made the concept of the energy barrier therein mathematically well-defined, it is related to the minimum energy cost the environment has to provide to the system in order to produce a generalized Pauli error, maximized for any generalized Pauli errors, not only logical operators. We evaluate this generalized energy barrier in Abelian quantum double models and find it to be a constant independent of system size. Thus, we rule out the possibility of entropic protection for this broad group of models.
\end{abstract}

\maketitle
\section{Introduction}
\label{sec:intro}

Whether it is possible to preserve arbitrary quantum information over a long period of time is a question of both fundamental and practical interest. \emph{Active} quantum error correction provides a way to protect quantum information but requires keeping track of and correcting the errors over a short time scale. Alternatively, quantum self-correcting systems would \emph{passively} preserve quantum information in the presence of a thermal environment without the need for external intervention on the system. The dynamics of these quantum "memories" would be such that the probability of an error occurring on the encoded information is exponentially suppressed with system size, resulting in an exponentially long memory time. Candidates for self-correction are typically systems governed by a local Hamiltonian whose degenerate ground space stores quantum information.

Assessing whether a system is self-correcting requires estimating the scaling of its memory time with system size. This difficult problem is often reduced to evaluating the \emph{energy barrier}, loosely defined as the maximal energy of intermediate states in a sequence of local transformations taking a ground state to an orthogonal ground state, minimized over all such possible sequences. This sequence of excited states mimics the evolution of the system under thermalization and decoding. The intuition (and implicit conjecture) is that the system obeys the phenomenological Arrhenius law which relates the memory time $t_\mathrm{mem}$ to the energy barrier  $\Delta E^*$ and the inverse temperature $\beta \equiv1/k_B T$ 
\begin{equation} \label{eq:Arrhenius_phenomenological}
t_\mathrm{mem}\propto e^{\beta \Delta E^*}
\end{equation}

The Arrhenius law is a useful guiding principle. For classical models, one can intuitively understand the exponentially long (classical) memory time of the ferromagnetic 2D Ising model by realizing that its energy barrier is proportional to the linear system size. Indeed, to go from the all up state to the all down state, one needs to flip a macroscopic droplet of spins whose energy scale with its perimeter. For quantum models, the most widely known example of a self-correcting quantum memory is the 4D Kitaev's toric code~\cite{Kitaev03, DKL+02, AHH+10} whose energy barrier is also proportional to the linear system size.  

The scaling energy barrier of a quantum model is intimately related to the geometrical support of operators mapping a ground state to a different orthogonal ground state, called \emph{logical operators}. For the 4D toric code, logical operators are tensor product of single qubit operators acting on a two-dimensional sheet-like subset of qubits, similar to the logical operator of the 2D Ising model which flips all spins. 

While the 4D Kitaev's toric code is self-correcting, it requires addressable long-range interactions if embedded in a lattice of lower dimensionality. Various attempts have been made to decrease the dimensionality of such a self-correcting code, while retaining a large energy barrier of the system \cite{Brell14,PHW+13,HCE+14}. A typical shortcomings of these codes includes sensitivity to perturbations \cite{LYP+15}. Finding a self-correcting system in three dimensions (or lower) is still an open question.

Following the intuition based on the Arrhenius law, it is believed that quantum self-correction requires a \emph{scaling} energy barrier, i.e., an energy barrier that is an increasing function of system size. However, a formal relation between self-correction and a scaling energy barrier has not been established and the Arrhenius law has only be proven for a few models while there are known counterexamples. Moreover, it was recently suggested that there might exist a different kind of protection \cite{LP13}, one that does not require a scaling energy barrier, coined \emph{entropy protection}. The intuition is that while there exist paths in phase space mapping a ground state to an orthogonal ground state while only introducing a constant amount go energy, these paths might not be the typical. Typical paths, however, might require the system to go through a scaling energy barrier. We could think of such a model as having an effective free energy barrier, i.e., there are free energy valleys in the landscape between the two ground states and in order to get out of such a valley the system would have to overcome an effective barrier. 

In 2014, Brown et al. proposed a local 2D Hamiltonian which seemed to realize entropy protection~\cite{BAP14} since its memory time exhibits a super exponential scaling, albeit only in a limited range of temperature. This model consists of a toric code-like structure, where instead of qubits $d$-level spins (qudits) are placed on the edges of a square lattice. This model also corresponds to the quantum double of $\mathbb{Z}_d$. Its elementary excitations are $d$ different electric and $d$ different magnetic anyons. Specifically, in Ref.~\cite{BAP14} $d=5$, and due to charge-flux duality, it is convenient to think only in terms of e.g. electric charges. Then there are $5$ different charges, grouped as: vacuum, light particle, heavy particle, heavy antiparticle, light antiparticle. Particle-antiparticle pairs have the same mass, furthermore the masses are set such that $m_{\textrm{heavy}} > 2 m_{\textrm{light}}$ to ensure that thermal evolution of the system favours the decay of a heavy particle into two light particles. The authors of Ref.~\cite{BAP14} further introduce defect lines to the system by modifying local terms of the Hamiltonian. When a light particle crosses such a line, it becomes a heavy one and vice versa. This construction results in fractal-like splitting of typical anyon-paths, resembling the fractal geometrical support of logical operators in Haah's cubic code \cite{Haah11,BH13}. The authors of Ref.~\cite{BAP14} numerically observed a memory time for this entropic code similar to the cubic code, that is, it grows super-exponentially with the inverse temperature ($t_{\textrm{mem}} \propto \exp(c\beta^2)$). 
A striking difference between the cubic code and Brown's entropic code is, however, that while the former has an energy barrier that grows logarithmically with system size, the energy barrier of the entropic code is a constant, independent of system size. Thus, Brown's entropic code seems to have a better scaling of memory time than the one predicted by the Arrhenius law.  However, it was also remarked that the super-exponential scaling did not remain valid at arbitrarily low temperature, i.e. in the limit of very large $\beta$. Thus, Brown's entropic code argues for the possibility of entropy protection but failed to settle the question whether entropy can protect quantum information and lead to a better scaling of than memory time than the one predicted by Arrhenius law.

Here, we settle this question in the negative by proving that a scaling energy barrier is necessary for self-correction for any quantum double model of an Abelian group, a general framework which contains Brown's entropic code. Thus, entropy cannot protect quantum information in the absence of a scaling energy barrier for those models. Technically, we establish a rigorous version of the Arrhenius law as an upper bound for the mixing time of quantum doubles of Abelian groups. We prove that the mixing time --defined as the longest time an initial state takes to thermalize to the Gibbs state-- and thus the memory time are upper bounded by $\textrm{poly}(N) \exp(2 \beta \overline{\epsilon})$ where $N$ is the size of the system and  $\overline{\epsilon}$ is the \emph{generalized energy barrier}. We rigorously define $\overline{\epsilon}$ by a natural quantity arising from our analysis which straightforwardly extends the intuitive notion of energy barrier. Finally, we evaluate  the generalized energy barrier and show that it is independent of system size or temperature for two-dimensional Abelian quantum double models. As our bound holds for any temperature, this means that Abelian quantum doubles don't allow for entropy protection, i.e., their memory time can at most scale exponentially with inverse temperature. Our results are based on the method presented in Ref.~\cite{Temme14} and are a generalization of the results therein, where the author has derived a similar Arrhenius law bound and energy barrier for any commuting Pauli stabilizer codes in any dimensions.

The paper is organized as follows. In Sec.~\ref{sec:quantum_doubles} and \ref{sec:noise_model} we introduce the framework of our analysis: the construction of Abelian quantum doubles and the noise model used to simulate the thermal environment. In Sec.~\ref{sec:energy_barrier} we present our main result: the upper bound on the mixing time and the formula for the generalized energy barrier, followed by a discussion on the physical interpretation of this result in Sec.~\ref{sec:Disc}. We present the details of the derivation of the bound in Sec.~\ref{sec:derivation}. Finally, we conclude with possible future directions in Sec.~\ref{sec:Concl}.


\section{Framework}

We now introduce the framework in which our result is valid. First, we introduce the systems of interests, i.e. the quantum double of Abelian groups. Second, we model the thermalization of such a system by the Davies map.

\subsection{Abelian quantum doubles}
\label{sec:quantum_doubles}

Abelian quantum doubles are a special case of the quantum double construction introduced by Kitaev \cite{Kitaev03}, where the quantum double is based on the cyclic group $\mathbb{Z}_d$. This was the model investigated in Ref.~\cite{BAP14} with $d=5$, and it is a generalized toric code construction (the toric code is the quantum double of $\mathbb{Z}_2$) acting on $d$-level spins or \emph{qudits}.

\subsubsection{Generalized Pauli operators}

We will choose a basis for the Hilbert space of a qudit to be labelled by orthonormal states $\left\{\ket{\ell}\right\}$ where $\ell\in\mathbb{Z}_d\cong\{0,\dots,d-1\}$.
We introduce the generalized Pauli operators, $X^k$ and $Z^k$, $k\in\mathbb{Z}_d$. They act on a qudit according to:
\begin{eqnarray}\label{eq:GenPauli}
X^k \left| \ell \right\rangle &=& \left| \ell\oplus k \right\rangle, \\
Z^k \left| \ell \right\rangle &=& \omega^{k\ell} \left| \ell \right\rangle ,
\end{eqnarray}
where $\oplus$ is the addition modulo $d$ and $\omega^\ell = \exp (i 2\pi \ell/d)$, $\ell\in\mathbb{Z}_d$ are the $d$th roots of unity. The eigenvalues of the $Z$ generalized Pauli operator but also the $X$ generalized Pauli operator are precisely the $d$th roots of unity. In our convention, the identity is a generalized Pauli operators with $k=0$. One can straightforwardly derive the following useful identities
\begin{equation}
X^\dagger=X^{d-1} \quad Z^\dagger=Z^{d-1} \quad Z^{k'}X^k=\omega^{kk'} X^k Z^{k'}.
\end{equation}

\subsubsection{Hamiltonian}

We now define the Hamiltonian of the quantum double of $\mathbb{Z}_d$ on $2N$ $d$-level spins or qudits located on the edge of a two dimensional square lattice with $N$ vertices. We define a (generalized) Pauli operator to be a $2N$-tensor product of single-qudit (generalized) Pauli operator $X^k$ or $Z^k$,  $k\in\mathbb{Z}_d$. For convenience, we will henceforth omit the (generalized) modifier. We note $\mathcal{P}_M$ the set of Pauli operators acting non-trivially on at most $M\leq 2N$ qudits. The qudits on which a Pauli operator acts non-trivially are its (geometrical) support. 

\begin{figure}
	\centering
		\includegraphics[width=0.22\textwidth]{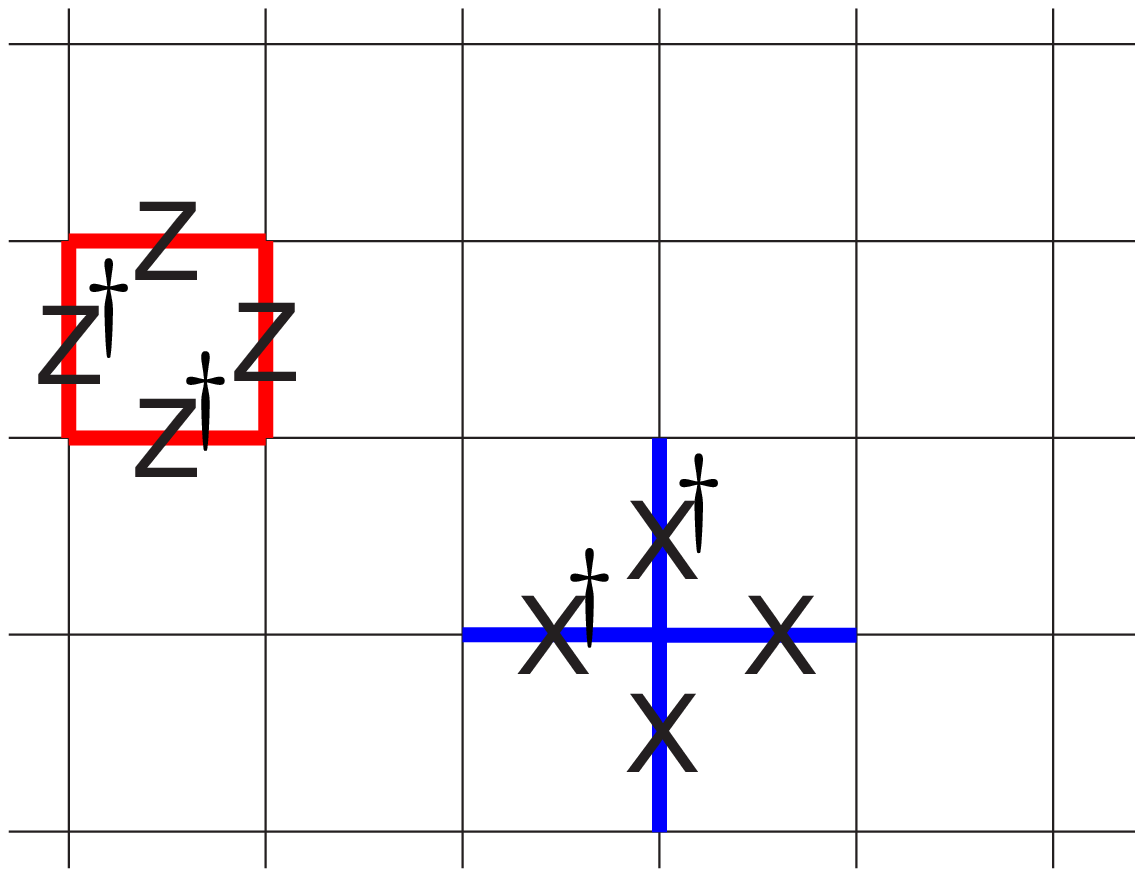}\hspace{0.2cm}
		\includegraphics[width=0.22\textwidth]{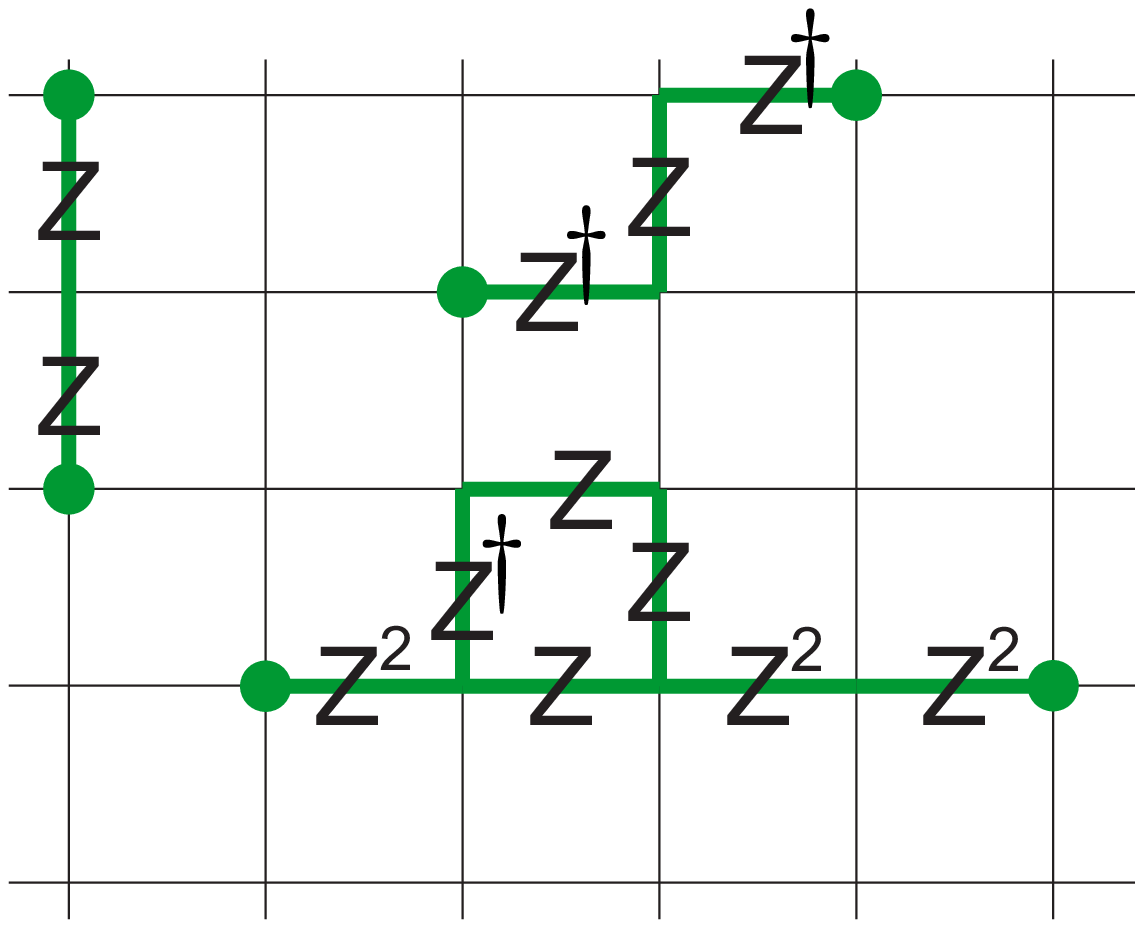}
		\caption{(a) Star operator $A(v)$ in blue and plaquette operator $B(p)$ in red. 
		(b) Examples of anyon paths for an Abelian quantum double ($d>2$).}
	\label{fig:operators}
\end{figure}

The local interactions of the Hamiltonian will be Pauli operators supported on four qudits neighbouring either a vertex $v$ of the lattice for star operators $A(v)$ or a plaquette $p$ for plaquette operators $B(p)$, see Fig.~\ref{fig:operators}(a). A star (and a plaquette) is the union of four edges or, equivalently, qudits located on those edges. It is convenient to label the qudits around a star $+$ or plaquette $\square$ using the cardinal points: East, South, West and North. The star operator $A(v)$ for vertex $v$ is 
\begin{equation}
A(v_i) = X_E \otimes X_S \otimes X^\dagger_W \otimes X^\dagger_N \quad (E,S,W,N)=+_v
\end{equation}
and the plaquette operator $B(p)$ for plaquette $p$ is 
\begin{equation}
B(p) = Z_E \otimes Z^\dagger_S  \otimes Z^\dagger_W \otimes Z_N \quad (E,S,W,N)=\square_p.
\end{equation}
The eigenvalues of star and plaquette operators are the $d$th roots of unity, inherited from the single-qudit Pauli operators. The projector unto the eigenvalue $\omega^a$ of the star operator at vertex $v$ is 
\begin{equation}
P^{a(v)}_v = \frac{1}{d} \sum_{k=0}^{d-1} (\omega^a A(v) )^k.
\end{equation}
Similarly, the projector unto the eigenvalue $\omega^a$ of the plaquette operator at plaquette $p$ is 
\begin{equation}
Q^{b(p)}_p = \frac{1}{d} \sum_{k=0}^{d-1} (\omega^b B(p) )^k.
\end{equation}
Note that those projectors commute since every star operator commute with every plaquette operator. 

The Hamiltonian of the $\mathbb{Z}_d$ quantum double is \cite{Kitaev03, BAP14}
\begin{equation}\label{eq:QDoubleHamiltonian}
H = \sum_v \sum_{a=0}^{d-1} J^{a(v)}_v P^{a(v)}_v + \sum_p \sum_{b=0}^{d-1} J^{b(p)}_p Q^{b(p)}_p ,
\end{equation}
where $J^{a(v)}_v$ and $J^{b(p)}_p$ are non-negative numbers. We set $\forall v,p \: J^0_v=J^0_p=0$ such that a ground state $\ket{\Omega}$ is a common +1 eigenvector of all $P^0_v$ and $Q^0_p$
\begin{equation}
\forall v,p \quad P^0_v \ket{\Omega}=Q^0_p \ket{\Omega}= + \ket{\Omega}.
\end{equation}
The ground space is degenerate whenever this Hamiltonian is defined on a manifold with non-zero genus. For instance, on a square lattice with periodic boundary condition, i.e., a torus, the ground space is $d^2$-degenerate and can be used to encode quantum information.  The positive numbers $J^{a(v)}_v$ and $J^{b(p)}_p$ for non-zero $a$ and $b$ can physically be interpreted as masses of the different excitations of the model, which we now discuss.  

\subsubsection{Excitations and syndromes}

Every spectral projector $P^{a(v)}_v$ and $Q^{b(p)}_p$ are pairwise commuting. Moreover, they commute with the Hamiltonian. Thus, it is convenient to label an energy eigenvector $\ket{\psi}$ using the quantum numbers $\mathbf{a}=\{a_v\}$ and $\mathbf{b}=\{b_p\}$ defined by
\begin{eqnarray}
a_v & = & \bra{\psi}A(v) \ket{\psi} \\
b_p & = & \bra{\psi}B(p) \ket{\psi}.
\end{eqnarray}
Using the terminology of quantum error correction, we define the \emph{syndrome} of $\ket{\psi}$ by
\begin{equation}
e(\ket{\psi})=(\mathbf{a},\mathbf{b}) \in \mathbb{Z}_d^{N+N}.
\end{equation}
Hence, the Hamiltonian can be diagonalized using the different syndrome values, i.e., 
\begin{equation}
H = \sum_{(\mathrm{a},\mathrm{b})} \epsilon (\mathbf{a},\mathbf{b}) \Pi (\mathbf{a},\mathbf{b}) ,
\end{equation}
where the explicit formula for the energies $\epsilon (\mathbf{a},\mathbf{b})$ and projectors $\Pi (\mathbf{a},\mathbf{b})$ can be found in Sec.~\ref{subsec:DiagHam_JumpOps}.   Let us try to draw a physical picture which will help intuition.

The syndrome of an energy eigenvector is a bookkeeping of the different excitations at every vertex and plaquette. The +1 eigenvectors of $P^{a(v)}_v$ for $a(v)\neq 0$ have a point-like excitation located on the vertex $v$ which we call an electric charge (or chargeon) of type $a$. Similarly, the +1 eigenvectors of $Q^{b(p)}_p$ for $b(p)\neq 0$ have a point-like excitation located on the plaquette $p$ which we call a magnetic flux (or fluxon) of type $b$. The ground states of the Hamiltonian have syndrome $(\mathbf{0},\mathbf{0})$. 

\begin{figure}
	\centering
		\includegraphics[width=0.15\textwidth]{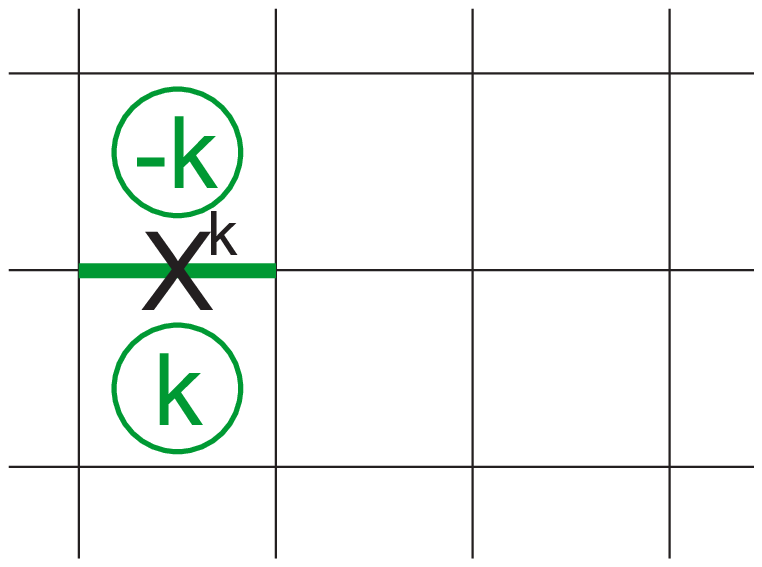}\hspace{0.1cm}
		\includegraphics[width=0.15\textwidth]{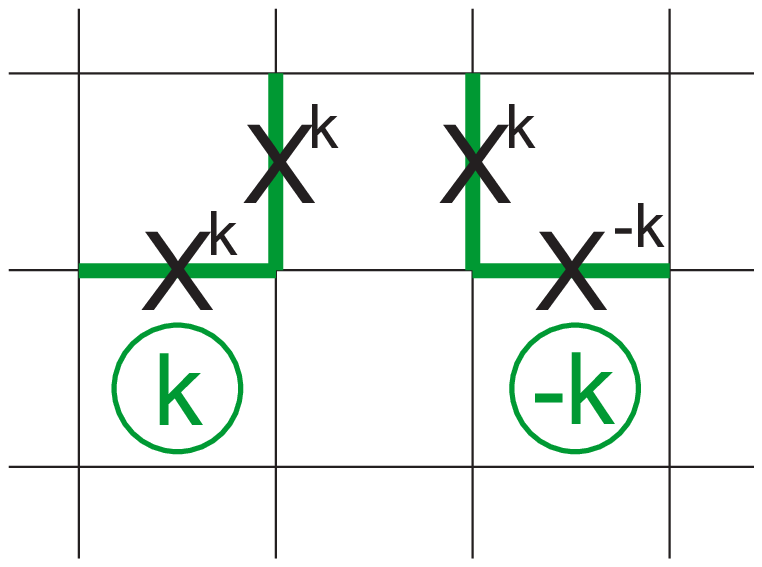}\hspace{0.1cm}
		\includegraphics[width=0.15\textwidth]{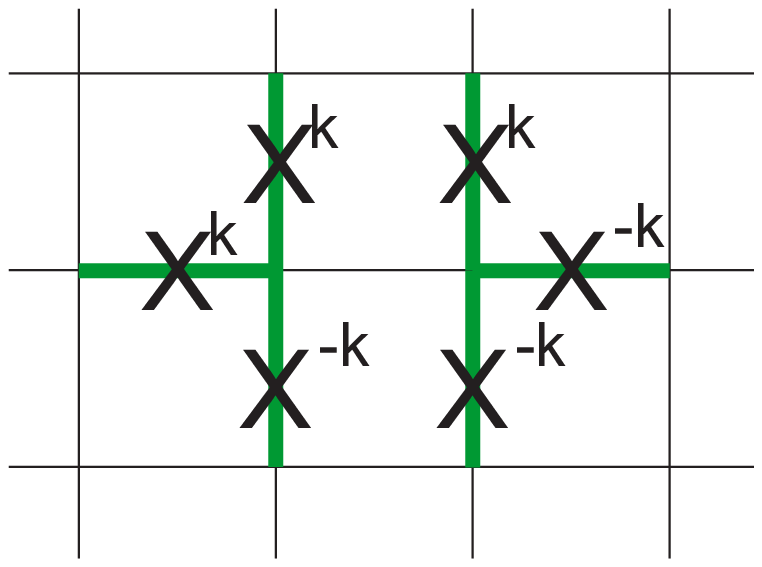}
		\caption{(a) Anyon pair created from vacuum, (b) one of the anyons moved and (c) the pair fused back to vacuum, all the while applying local operators.}
		\label{fig:anyon_creation}
\end{figure}

Physically, the point-like excitations can (i) be created out of the vacuum by applying a local operator on a ground state, (ii) propagate on the lattice and (iii) annihilate back to the vacuum by applying a local operator. This can be understood at the level of the syndrome. Consider on a ground state and then apply a generalized Pauli operator $X^k$ on a qudit located on a horizontal edge (see Fig.~\ref{fig:anyon_creation}).  This will modify the eigenvalues of the plaquette operators North and South of that horizontal edge, denoted $B(p_N)$ and $B(p_S)$. Indeed, the resulting state will be a +1 eigenvector of the spectral projectors $Q^{-k}_{p_N}$ and $Q^{k}_{p_S}$. Physically, $X^k$ created a magnetic flux of type $k$ (resp. $-k$) on the South (resp. North) plaquettes. In other words, $X^k$ created a pair of conjugate magnetic fluxes out of the vacuum. Similarly, $Z^k$ would create a pair of conjugate electric charges out of the vacuum. 
We assign to any generalized Pauli operator $\sigma_{\boldsymbol{\eta}}$ the syndrome $\mathbf{e}(\boldsymbol{\eta})$ of the state $\sigma_{\boldsymbol{\eta}} \ket{\Omega}$
\begin{equation}
\mathbf{e}(\boldsymbol{\eta})=\mathbf{e}(\sigma_{\boldsymbol{\eta}} \ket{\Omega}).
\end{equation}
In our examples, 
\begin{eqnarray}
e(X^k) & = & (\mathbf{a}=\mathbf{0}, \mathbf{b}=\left[0,\dots,0,k,-k,0,\dots,0\right]) \\
e(Z^k) & = & (\mathbf{a}=\left[0,\dots,0,k,-k,0,\dots,0\right] , \mathbf{b}= \mathbf{0}). 
\end{eqnarray}
Given any energy eigenvector $\ket{\psi}$ and any generalized Pauli operator $\boldsymbol{\eta}$, the syndrome of the state $\sigma_{\boldsymbol{\eta}} \ket{\psi}$ is obtained by 
\begin{equation}
\mathbf{e}(\sigma_{\boldsymbol{\eta}} \ket{\psi}) = \mathbf{e}(\boldsymbol{\eta}) \oplus \mathbf{e}(\ket{\psi}).
\end{equation}
This very simple addition rule stems for the Abelian structure of the group $\mathbb{Z}_d$ and is  related to the fusion rules of the excitations of this Abelian topological model.

Fluxons and chargeons turn out to be (Abelian) anyons, i.e., quasi-particles which are not bosonic nor fermionic. Yet their anyonic nature will not be essential in our work. However, we will from now on use the term anyon to designate a generic point-like excitation (either a fluxon or a chargeon).  Moreover, chargeons and fluxons are related by an exact duality which maps the lattice to the dual lattice. Thus, it will often be convenient to focus on a single anyon type, e.g., chargeons in order to simplify our discussion and notations. Also, we would like to introduce a single index $s$ which labels either the vertices or the plaquettes, i.e., $s=v/p$. Thus, any anyon (chargeon or fluxon) is located on a site (vertex or plaquette).


\subsection{Thermal noise model}
\label{sec:noise_model}

The model used in our work to simulate the thermalization process of the quantum double is the Davies map \cite{Davies76, Davies79}, the gold standard for simulating the thermalization of many body systems \cite{BLP+14, AHH+10, BH13, BAP14}. The system is coupled to a bosonic bath and the  Hamiltonian of \{system+bath\} reads
\begin{equation}
H_{\textrm{full}} = H_{\textrm{system}} + \chi \sum_{\alpha}{ S_{\alpha} \otimes B_{\alpha}} + H_{\textrm{bath}},
\end{equation}
where $B_\alpha$ is the operator acting on the bath and $S_\alpha\equiv S^j_{\alpha'}$ is an operator acting on spin $j$ of the system. We consider the weak coupling limit, with $\chi\ll 1$.

The density operator of the system, noted $\rho$, evolves according to the master equation
\begin{equation}
\frac{d \rho}{dt} = - i [H_{\textrm{eff}}, \rho] + \mathcal{L}(\rho) ,
\end{equation}
where $H_{\textrm{eff}}$ is the (Lamb-shifted) system Hamiltonian 
$H_{\textrm{eff}} = H_{\textrm{system}} + \sum_{\alpha, \omega} S^\dagger_\alpha (\omega) S_{\alpha} (\omega)$
and the Liouvillian is
\begin{equation}\label{eq:DaviesGen}
\mathcal{L}(\rho) = \sum_{\alpha, \omega}  \gamma_\alpha (\omega) \left( S_{\alpha} (\omega) \rho S^\dagger_\alpha (\omega) - \frac{1}{2} \{ S^\dagger_\alpha (\omega) S_{\alpha} (\omega) , \rho \}_+ \right).
\end{equation}
The operators governing the evolution of the system in energy space are the spectral jump operators, $S_{\alpha} (\omega)$. They take the system from energy eigenstate $\epsilon'$ to another eigenstate with energy $\epsilon = \epsilon' + \omega$, and have the form
\begin{equation} \label{eq:jump_operator_omega}
S_{\alpha} (\omega) = \sum_{\epsilon(\mathbf{a},\mathbf{b}) - \epsilon(\mathbf{a'},\mathbf{b'}) = \omega} \Pi(\mathbf{a},\mathbf{b}) S_\alpha \Pi(\mathbf{a'},\mathbf{b'}) .
\end{equation} 
They are the Fourier transforms of $S_{\alpha} (t)$ (the time-dependent operator acting on the system due to its contact with the thermal bath):
\begin{equation} \label{eq:jump_operator_def}
S_{\alpha} (t) = \sum_{(\mathbf{a},\mathbf{b}),(\mathbf{a'},\mathbf{b'})} e^{i \epsilon(\mathbf{a},\mathbf{b}) t} \Pi(\mathbf{a},\mathbf{b}) S_\alpha \Pi(\mathbf{a'},\mathbf{b'}) e^{- i \epsilon(\mathbf{a'},\mathbf{b'}) t}.
\end{equation}
The rate with which a state of the system is taken to another state $\omega$ far in energy, by applying the jump operator $S_\alpha (\omega)$ due to its coupling to the thermal bath is the transition rate $\gamma_\alpha(\omega)$. These transition rates obey detailed balance
\begin{equation}
\gamma_\alpha(\omega)=e^{\beta \omega} \gamma_\alpha(-\omega).
\end{equation}

This Liouvillian drives any state towards the Gibbs state
\begin{equation} \label{eq:Gibbs_state}
\rho_G\propto e^{-\beta H_\textrm{system}}
\end{equation}
which is its unique fixed point: $\mathcal{L}(\rho_G) = 0$.

Applying the Davies map to a $\mathbb{Z}_d$ quantum double we need to choose an operator basis for the jump operators $S_\alpha$. For $d=2$, a possible choice is the Pauli group, while for $d>2$ it is the generalized Pauli group. We should be careful, since although the elements of the Pauli group are Hermitian, the elements of the generalized group are not: $X^\dagger = X^{d-1}$. We can circumvent this problem by either writing the interaction terms in the full Hamiltonian as $\sigma_{j,\alpha'} \otimes B_{\alpha}^\dagger + \sigma_{j,\alpha'}^\dagger \otimes B_{\alpha}$ with $\sigma_{j,\alpha'=(l,m)}=Z_j^l X_j^m$, thus $S_{j,\alpha'} = Z_j^l X_j^m$ as in the $\mathbb{Z}_2$ case, or by constructing Hermitian jump operators: $S_{j,\alpha'} = 1/\sqrt{2} (\sigma_{j,\alpha'}+\sigma_{j,\alpha'}^\dagger)$. Independent of which choice we make, our results in the following sections are the same.


\section{Generalized energy barrier}
\label{sec:energy_barrier}

We establish a formerly ill-defined link between the energy barrier of a system and its mixing time for Abelian quantum doubles. We prove a rigorous Arrhenius law upper bound for the mixing time (Sec.~\ref{subsec:Arrh_law_bound}, details of the proof in section~\ref{sec:derivation}), and give a proper definition for the energy barrier appearing in that bound (Sec.~\ref{subsec:en_barrier_def}). In section~\ref{subsec:en_barrier_construction} we evaluate this energy barrier for Abelian quantum doubles in two dimensions and find it is a constant independent of system size or temperature.

\subsection{Definition of the generalized energy barrier}\label{subsec:en_barrier_def}

Recall from the introduction that the energy barrier is intuitively related to the decomposition of operators acting non-trivially within the ground space (logical operators) into a sequence of local operators. Surprisingly, the \emph{generalized} energy barrier arising from our analysis is related to the energy cost of building an \emph{arbitrary} Pauli operator. This seems to go against intuition since an arbitrary Pauli operator $\sigma_{\boldsymbol{\eta}}\in\mathcal{P}_{2N}$ can create an extensive amount of energy. However, excitations which appear in the final error configuration $\mathbf{e}(\boldsymbol{\eta})$ created by the Pauli operator will not contribute towards the generalized energy barrier: only intermediate excitations created in the sequential construction of this final error configuration do. Note that if $\boldsymbol{\eta}$ is a logical operator, the generalized energy barrier coincides with the intuitive energy barrier.

The idea is thus to consider sequences of Pauli operators $\left\{ \sigma_{\overline{\boldsymbol{\eta}}^t} \right\}$ which sequentially build the operator $\sigma_{\boldsymbol{\eta}}$ by applying Pauli operators acting on a single qudit. We call such a sequence a \emph{local errors path}. Indeed, we think of $\boldsymbol{\eta}$ as the index of the final error which we sequentially build through single qudit errors such that the error at step $t$ is indexed by $\overline{\boldsymbol{\eta}}^{t}$. 

\begin{definition}[Local errors path]\label{def:local_err_path}
A local errors path $\{\sigma_{\overline{\boldsymbol{\eta}}^t}\}_{t\geq 0}$ is a sequence of Pauli operators such that
\begin{align}
\sigma_{\overline{\boldsymbol{\eta}}^{t=0}}=\mathbb{I} \\
\textrm{locality} \qquad \forall t \: \exists P\in\mathcal{P}_1 \quad \sigma_{\overline{\boldsymbol{\eta}}^{t+1}}=P \cdot \sigma_{\overline{\boldsymbol{\eta}}^{t}} \\
\textrm{convergence} \qquad \exists \: \sigma_{\boldsymbol{\eta}},T  \quad t>T \Rightarrow \sigma_{\overline{\boldsymbol{\eta}}^{t}}=\sigma_{\boldsymbol{\eta}} 
\end{align}
\end{definition}

At any intermediate step $t\leq T$, the Pauli operator $\sigma_{\overline{\boldsymbol{\eta}}^{t}}$ will create a syndrome $\mathbf{e}(\overline{\boldsymbol{\eta}}^{t})$ corresponding to a pattern of anyons. At every site, only the energy of an anyon whose charge is different from the one in the syndrome $\mathbf{e}(\boldsymbol{\eta})$ contribute towards the energy barrier. Formally, we define the additional energy of the error indexed by $\overline{\boldsymbol{\eta}}^{t}$ with respect to the error indexed by $\boldsymbol{\eta}$ as

\begin{definition}[Additional energy] \label{def:additional_energy}
Let $\overline{\boldsymbol{\eta}}^{t}$ and $\boldsymbol{\eta}$ be indices of two Pauli operators. The additional energy of the error  $\sigma_{\overline{\boldsymbol{\eta}}^{t}}$ with respect to the reference operator $\sigma_{\boldsymbol{\eta}}$ is

\begin{equation} \label{eq:additional_energy}
\overline{\epsilon} (\overline{\boldsymbol{\eta}}^{t}|\boldsymbol{\eta})=\sum_s J_s^{e_s(\overline{\boldsymbol{\eta}}^{t})} \left(1-\delta_{\mathbf{e}_s(\overline{\boldsymbol{\eta}}^t),0}\right) \left(1-\delta_{e_s(\overline{\boldsymbol{\eta}}^t),e_s(\boldsymbol{\eta})}\right)
\end{equation}

\end{definition}

Note that in Eq.~\eqref{eq:additional_energy}, summands do not contribute if $e_s(\overline{\boldsymbol{\eta}}^{t})=0$, i.e., if the intermediate error does not create excitations on site $s$ but also if $e_s(\overline{\boldsymbol{\eta}}^t)=e_s(\boldsymbol{\eta})$, i.e., if the intermediate error creates the same excitation on site $s$ as the reference error $\sigma_{\boldsymbol{\eta}}$.

We are now in position to define the generalized energy barrier of an error $\sigma_{\boldsymbol{\eta}}$ and then of the Hamiltonian.

\begin{definition}[Generalized energy barrier] 
Let $\sigma_{\boldsymbol{\eta}}\in\mathcal{P}_{2N}$ be a Pauli operator and $\{\sigma_{\overline{\boldsymbol{\eta}}^t}\}$ denote an arbitrary local errors path converging to $\sigma_{\boldsymbol{\eta}}$. The generalized energy barrier of $\sigma_{\boldsymbol{\eta}}$ is
\begin{equation}\label{eq:energy_barrier_operator}
\overline{\epsilon} (\boldsymbol{\eta}) = \min_{\{\sigma_{\overline{\boldsymbol{\eta}}^t}\}\to \sigma_{\boldsymbol{\eta}}} \max_t \overline{\epsilon} (\overline{\boldsymbol{\eta}}^{t}|\boldsymbol{\eta})
\end{equation}

The generalized energy barrier of the Hamiltonian H is
\begin{equation}\label{eq:energybarrier}
\overline{\epsilon} (H) = \max_\eta \overline{\epsilon} (\boldsymbol{\eta}) .
\end{equation}

\end{definition}

We now introduce the mixing time, an upper bound on the quantum memory time, and then introduce our bound which relates it to the generalized energy barrier through a formula similar to the Arrhenius law given in Eq.~\eqref{eq:Arrhenius_phenomenological}.

\subsection{Arrhenius upper bound on the mixing time}\label{subsec:Arrh_law_bound}
 
We define the mixing time as the time scale after which the evolution of any initial state of the system becomes $\varepsilon = e^{-1/2}$-indistinguishable from the Gibbs state defined by Eq.~\eqref{eq:Gibbs_state}. The $\varepsilon = e^{-1/2}$ value is chosen so the relationship between the mixing time and the gap of the Liouvillian will have a convenient form, and the exact value won't modify either the qualitative aspect of our calculations or the scaling of the bound obtained on the mixing time.

\begin{definition}[($\varepsilon$)-mixing time]
The ($\varepsilon$)-mixing time of a Liouvillian (whose fixed point is the Gibbs state $\rho_G$) is 
\begin{equation}
t_{\textrm{mix}} (\varepsilon) = \min \{ t \ | \ t'>t \Rightarrow  || e^{\mathcal{L}t'} \rho_0 - \rho_G ||_1 < \varepsilon \:  \forall \rho_0 \} .
\end{equation}
\end{definition}
where we used the trace norm, $||A||_1=\textrm{Tr}\left[ \sqrt{A^\dagger A} \right]$, to measure the (in)distinguishability of two quantum states.

\begin{definition}[Mixing time]
The mixing time of a Liouvillian is its $\varepsilon = e^{-1/2}$ ($\varepsilon$)-mixing time.
\end{definition}

Loosely defining the quantum memory time as the maximal time after which one can recover information about the initial ground state, we immediately see it is upper bounded by the mixing time. Indeed, the Gibbs state treats all ground state on the same footing and thus information about the initial ground state has disappeared. We do not provide a formal definition of the quantum memory time in this work.

Our main result relates the generalized energy barrier to the mixing time through a relation similar to the Arrhenius law.

\begin{theorem}[Arrhenius bound on mixing time]
For any Abelian group $\mathbb{Z}_d$, for any inverse temperature $\beta$, the mixing time of the Davies map Liouvillian of the quantum double of $\mathbb{Z}_d$ is upper bounded by 

\begin{equation} \label{eq:mixing_time_bound}
t_{\textrm{mix}} \leq \mathcal{O} \left( \beta N \mu(N) e^{\beta (2 \bar{\epsilon} + \Delta)} \right) ,
\end{equation}
where $2N$ is the number of qudits in the system, $\Delta$ is the gap of the system Hamiltonian, $\bar{\epsilon}$ is the generalized energy barrier and $\mu(N)$ defined by Eq.~\eqref{eq:def_mu} is the length of the longest optimal local errors path. 
\end{theorem}

The derivation of this result can be found in section \ref{sec:derivation}. We will now show that for Abelian quantum double, $\bar{\epsilon}$ is bounded by a constant independent of system size in Sec.~\ref{subsec:en_barrier_construction} and that $\mu(N)$ is bounded by $8N(d-1)$ in Sec.~\ref{sec:canPathLength}. The right hand side of Eq.~\eqref{eq:mixing_time_bound} has a dependence on a low power of $N$, that does not qualitatively modify the scaling of the mixing time nor the behaviour of the system when considered as a candidate for a quantum memory. The important physical quality of this bound is the Arrhenius law scaling. This scaling is set by the gap of the system Hamiltonian but, more interestingly, by the generalized energy barrier, which we now evaluate.

\subsection{The generalized energy barrier is a constant for Abelian quantum doubles}
\label{subsec:en_barrier_construction}

We will now evaluate the generalized energy barrier of any 2D quantum double of an Abelian group and show that it is a constant, independent of system size, more precisely $2 J_{\textrm{max}}$. While this was known for the $\mathbb{Z}_2$ case~\cite{Temme14}, we extend it to any $\mathbb{Z}_d$ quantum double. From now on, we consider a $\mathbb{Z}_d$ quantum double, with arbitrary $d$. Furthermore, we henceforth omit the 'generalized' modifier in  (generalized) energy barrier for simplicity.

To evaluate the energy barrier of the Hamiltonian, given by Eq.~\eqref{eq:energybarrier}, we want to bound the barrier of an arbitrary Pauli operator, given by Eq.~\eqref{eq:energy_barrier_operator}. Thus, we aim to exhibit a local errors path where the additional energy of any intermediate error is a constant. To do so, we will use the following strategy. We will first turn the final error syndrome into a weighted directed graph intuitively corresponding to the worldlines of anyons. Then, we will decompose that graph into cycles and trees. Cycles correspond to pair of conjugate anyons appearing out of the vacuum, propagating and then fusing back to the vacuum. Trees represent propagation of anyons, whose position (resp. word lines) correspond to terminal vertex (resp. edges) of the tree. Finally, using different techniques for cycles and trees, we show how to build the error of each type by moving at most one anyon at a time in a way that the additional energy of any intermediate error involve at most two local terms of the Hamiltonian, resulting in an energy barrier of at most $2J_\textrm{max}$. 

\subsubsection{Graph corresponding to an error configuration}

\begin{itemize}

\item \textit{Any error is the product of elementary errors whose supports are disjoint. The energy barrier of the error is the largest energy barrier of its elementary errors.}

Let's consider an arbitrary error, i.e., a Pauli operator. Its geometrical support, i.e., qudits on which it acts non-trivially, splits into connected components. We can decompose the global operator into a product of elementary operators, each of which is supported on one connected component. No terms of the Hamiltonian has support which intersect two connected components. Thus, we can choose the local errors path so that elementary errors are built sequentially. In any intermediate error, there is a unique elementary error under construction. The other elementary errors are either not constructed yet, or already constructed. In either case, they do not contribute towards the energy barrier.

\item \textit{Any elementary error can be interpreted as anyons decaying and fusing together, thus forming a fully connected, directed graph with weighted edges. Due to charge conservation, this graph is a flow.}

The error is a tensor product of single-qudit Pauli operators which create pair of conjugate anyons out of the vacuum, fuse and move anyons. Thus, to every elementary error, we can associate a graph with weighted edges, often referred to as string-nets in the literature. Such a graph is depicted on Fig.~\ref{fig:lemmas}. 

A terminal vertex, i.e., a vertex of valency 1, is an anyon. It is convenient to label terminal vertices by their anyonic charge, i.e., the value of the syndrome of the elementary error on that site. Other vertices correspond to world lines of anyons. Vertices are linked by an edge of weight $k$ if the errors between them correspond to moving an anyon of charge $k$ along the orientation of the edge. An edge of weight $k$ connecting site $i$ to site $j$ is equivalent to an edge of weight $d-k$ connecting $j$ to $i$. 

At this point, we have built a directed graph satisfying weight conservation at every vertex.  Indeed, weight conservation in the graph is equivalent to charge conservation of the anyons in this Abelian topological model. Such a graph is called a graph flow. 

\end{itemize}

\begin{figure}
	\centering
		\includegraphics[width=0.22\textwidth]{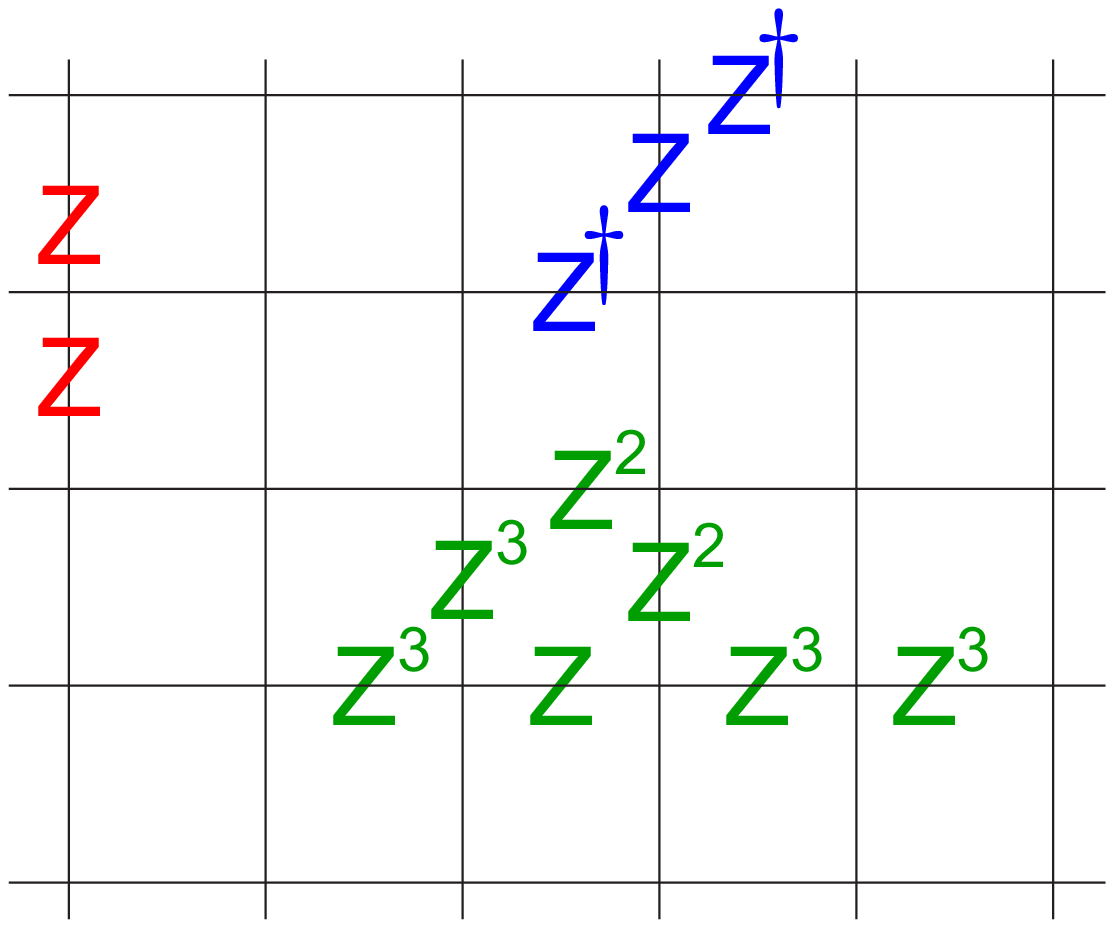}\hspace{0.2cm}
		\includegraphics[width=0.22\textwidth]{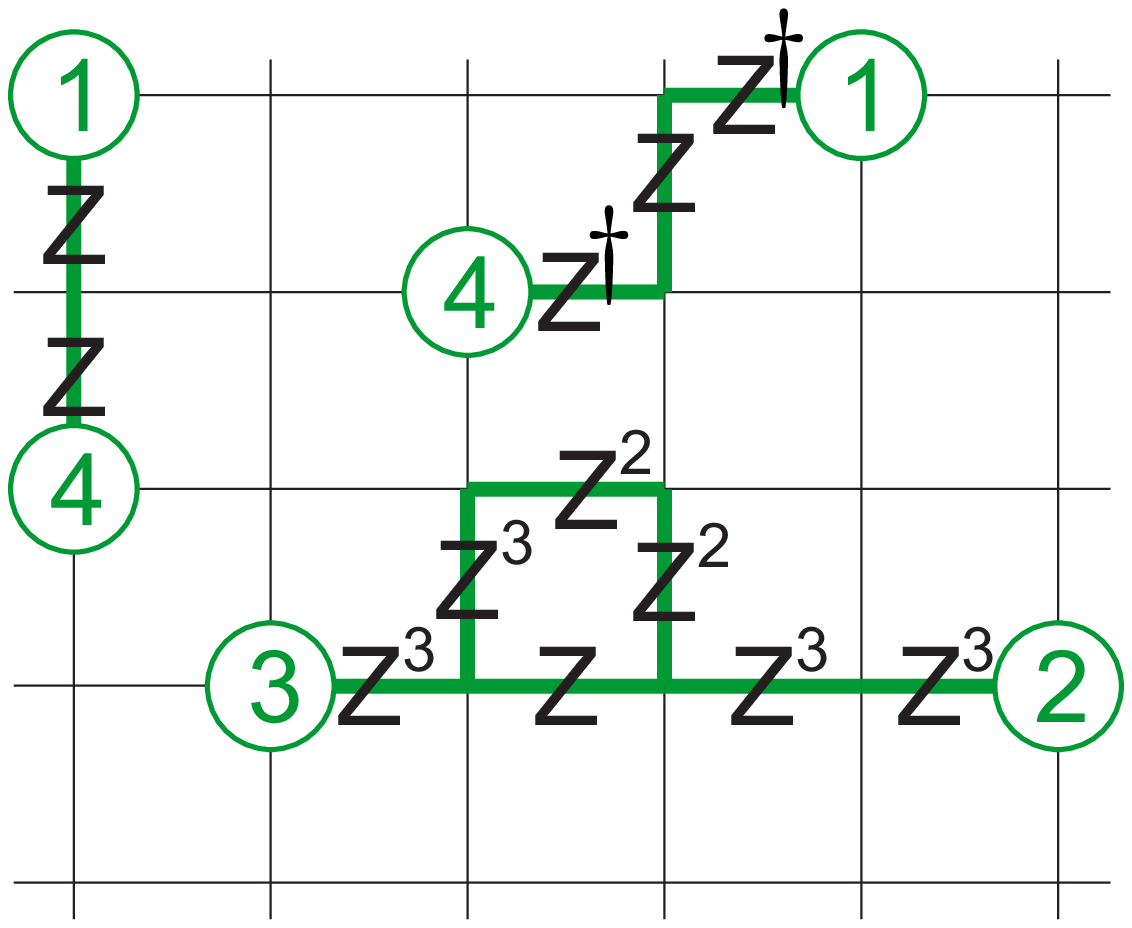}\\ \vspace{0.2cm}
		\includegraphics[width=0.22\textwidth]{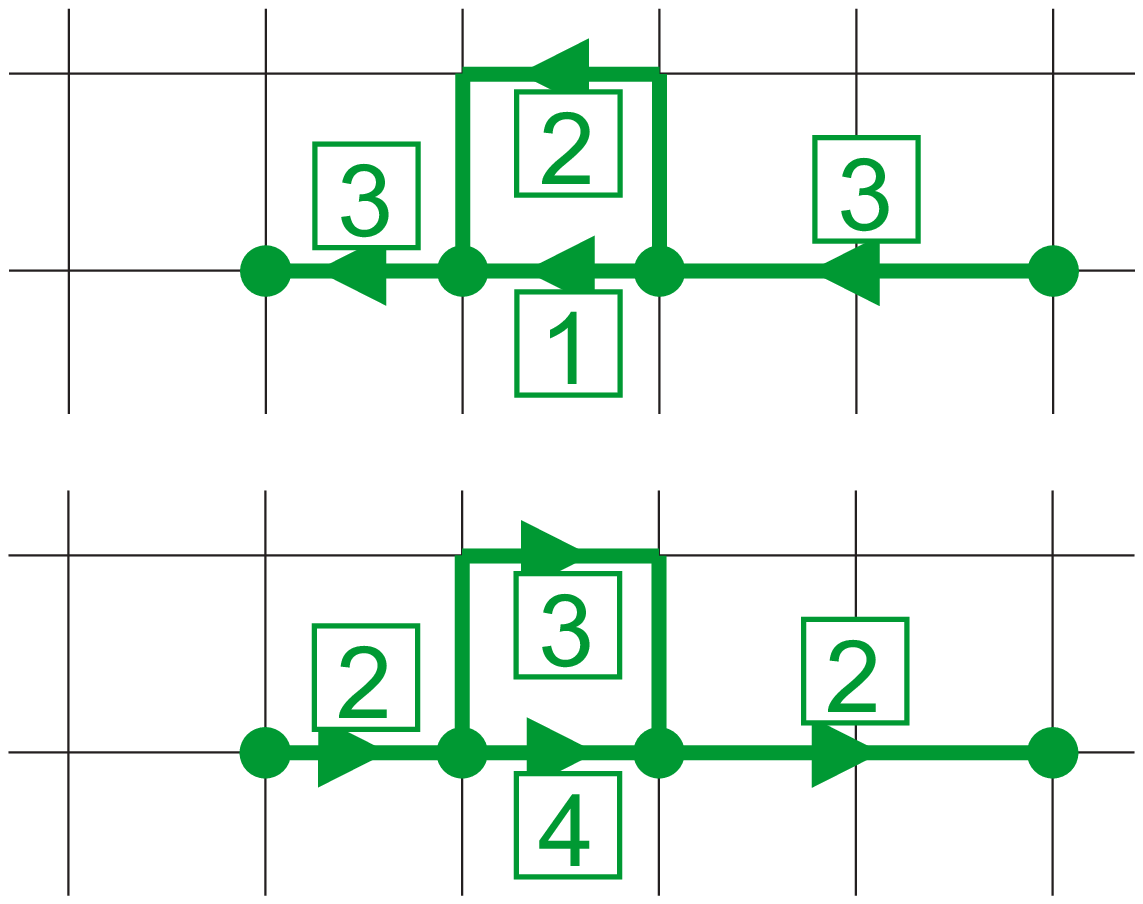}\hspace{0.2cm}
		\includegraphics[width=0.22\textwidth]{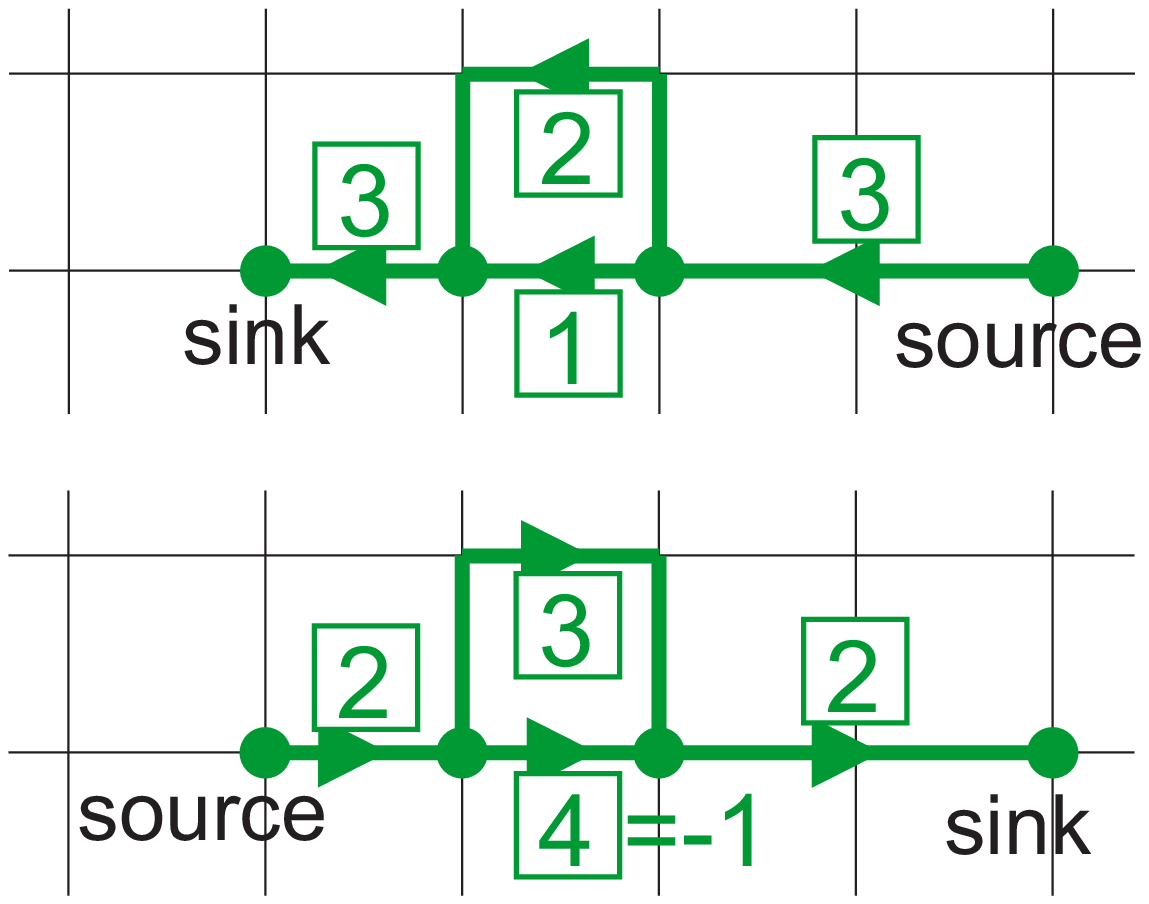}\\ \vspace{0.2cm}
		\includegraphics[width=0.22\textwidth]{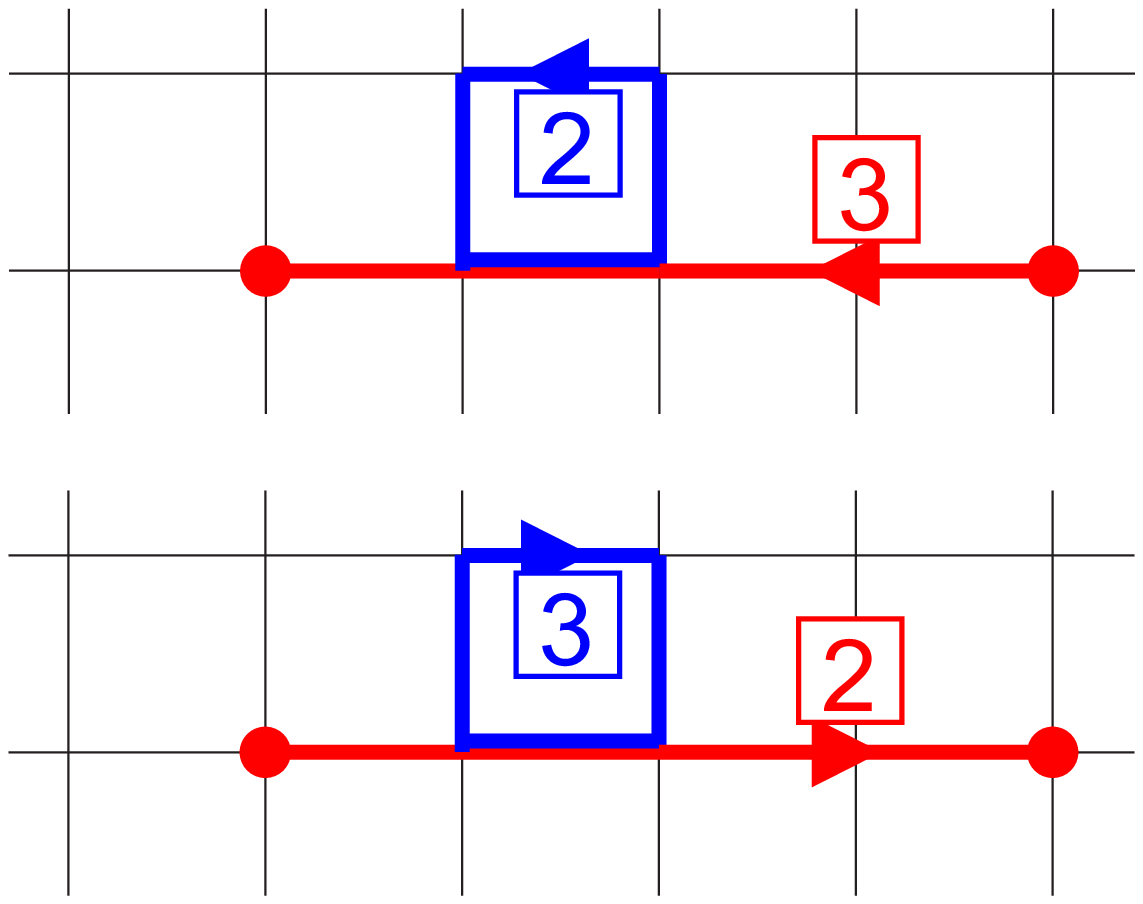}\hspace{0.2cm}
		\includegraphics[width=0.22\textwidth]{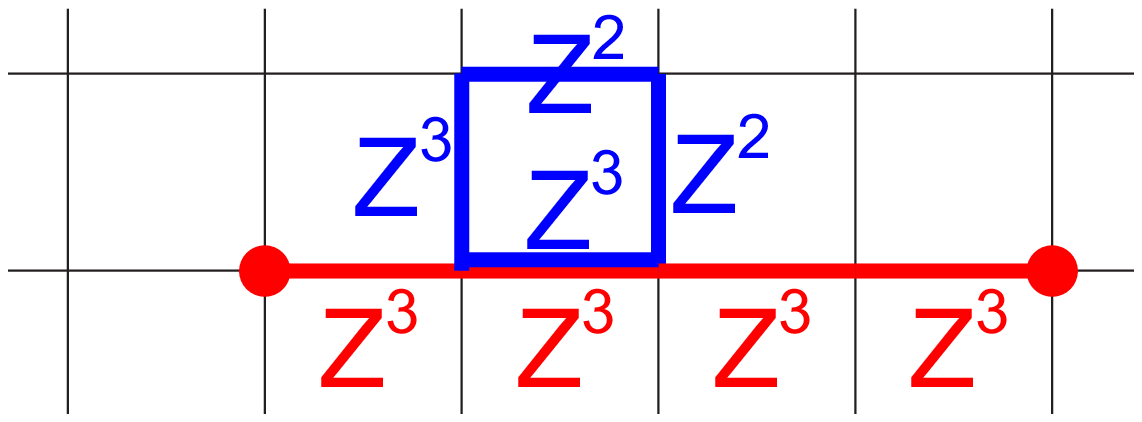}
		\caption{Illustration of the steps towards constructing the optimal canonical path for quantum double $\mathbb{Z}_5$: (a) the support of the error can be partitioned into three connected components, coloured in red, blue and green; (b) each elementary error can be interpreted as fusions, decays and moving of anyons; (c) an elementary error can be mapped onto a fully connected directed graph with weighted edges (reversing the orientation of an edge changes the weight from $k$ to $d-k$); (d) this can be interpreted as a flow of charges; (e) the flow can be partitioned into a rotational (blue) and irrotational (red) part; and (f) these different partitionings of the flow are all equivalent to applying the same combination of operators in either loops or strings.}
	\label{fig:lemmas}
\end{figure}

\subsubsection{Decomposing the graph into cycles and trees}

We now use a well-known result from flow theory: any flow can be partitioned into three sets: a rotational and an irrotational flow and a harmonic component (Helmholtz-Hodge decomposition) \cite{Frankel11}.

On this discrete geometry of a graph, the rotational flow consist of loops (a.k.a cycles), the irrotational flow consists of trees which can be thought as union of strings and the harmonic part consist of irrotational flows on the non-contractible cycles.

This decomposition can be physical interpreted in terms of anyons which we now do in order to evaluate the energy barrier. Remember the rules of the additional energy, defined in Definition.~\ref{def:additional_energy} : an intermediate error has additional energy if an anyon at a given site is not the anyon created by the reference error. The goal is now to build the reference error by introducing as little additional energy as possible. 

\subsubsection{Evaluating the energy barrier} 
\label{sec:energy_barrier_evaluation}

Loops correspond to a particle-antiparticle pair appearing out of the vacuum, then propagating and eventually fusing back to the vacuum. Such a configuration can be created by moving two anyons.  Thus, loops have an energy barrier corresponding to the energy of two anyons.

We now explain how to construct an error whose support is a tree. 

We can consider the tree to be a superposition of strings, each string corresponding to a pair of conjugate anyons which has been created out of the vacuum and then propagated. The terminal vertex of the tree correspond to anyons, conveniently labelled by their anyonic charge. For convenience, choose one of these terminal vertices to be the root of the tree. Other terminal vertices will now be called "leaves". The root is connected to each leaf by a path whose weight is the anyonic charge of the leaf. Each such path is a string operator connecting an anyon (at the leaf) to its conjugate anyon (at the root). See Fig.~\ref{fig:tree2strings} for a graphical example.

\begin{figure} 
	\centering
		\includegraphics[angle=90,origin=c,width= \columnwidth]{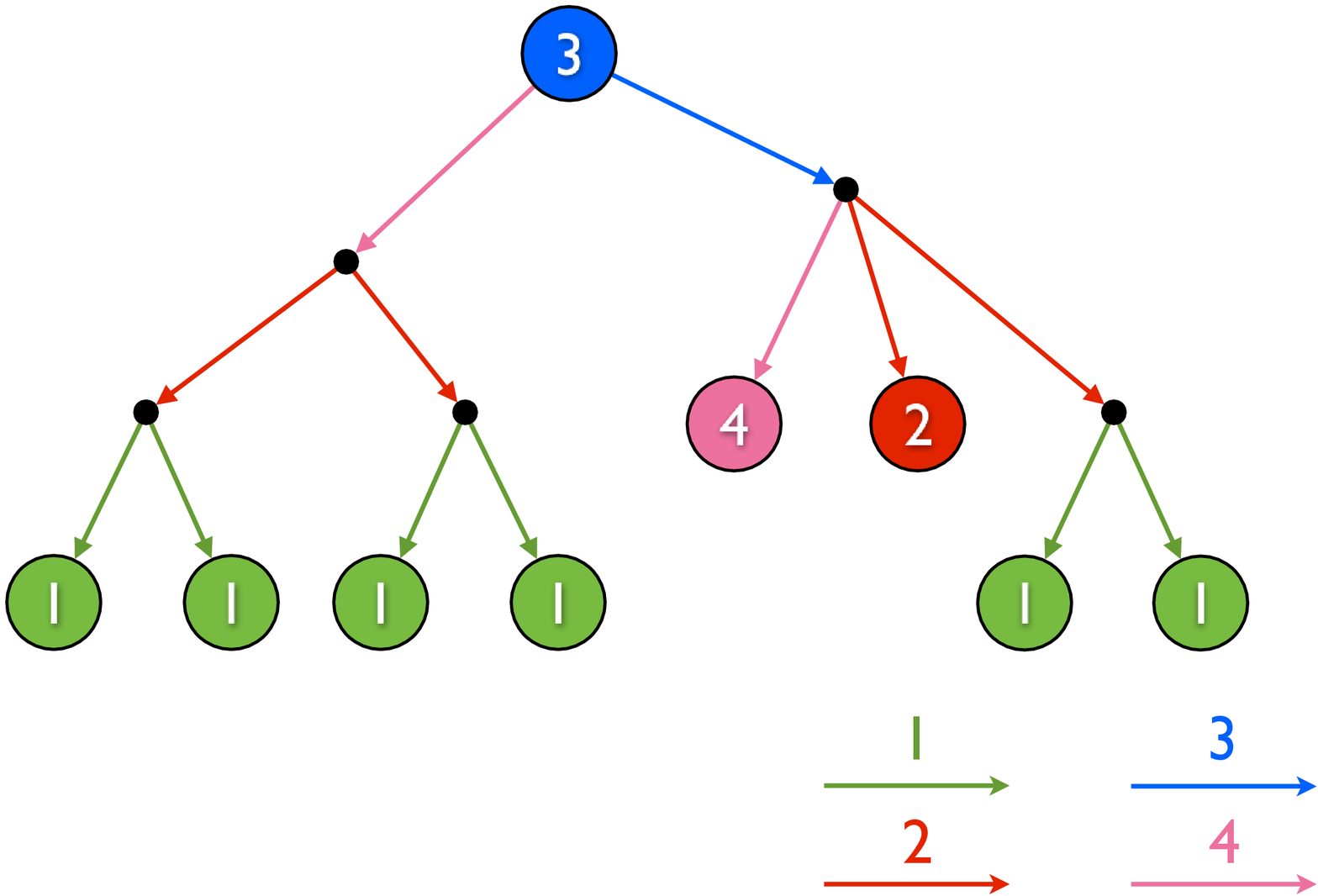} 
		
		\vspace{-1cm}
				
		\includegraphics[angle=90,origin=c,width= \columnwidth]{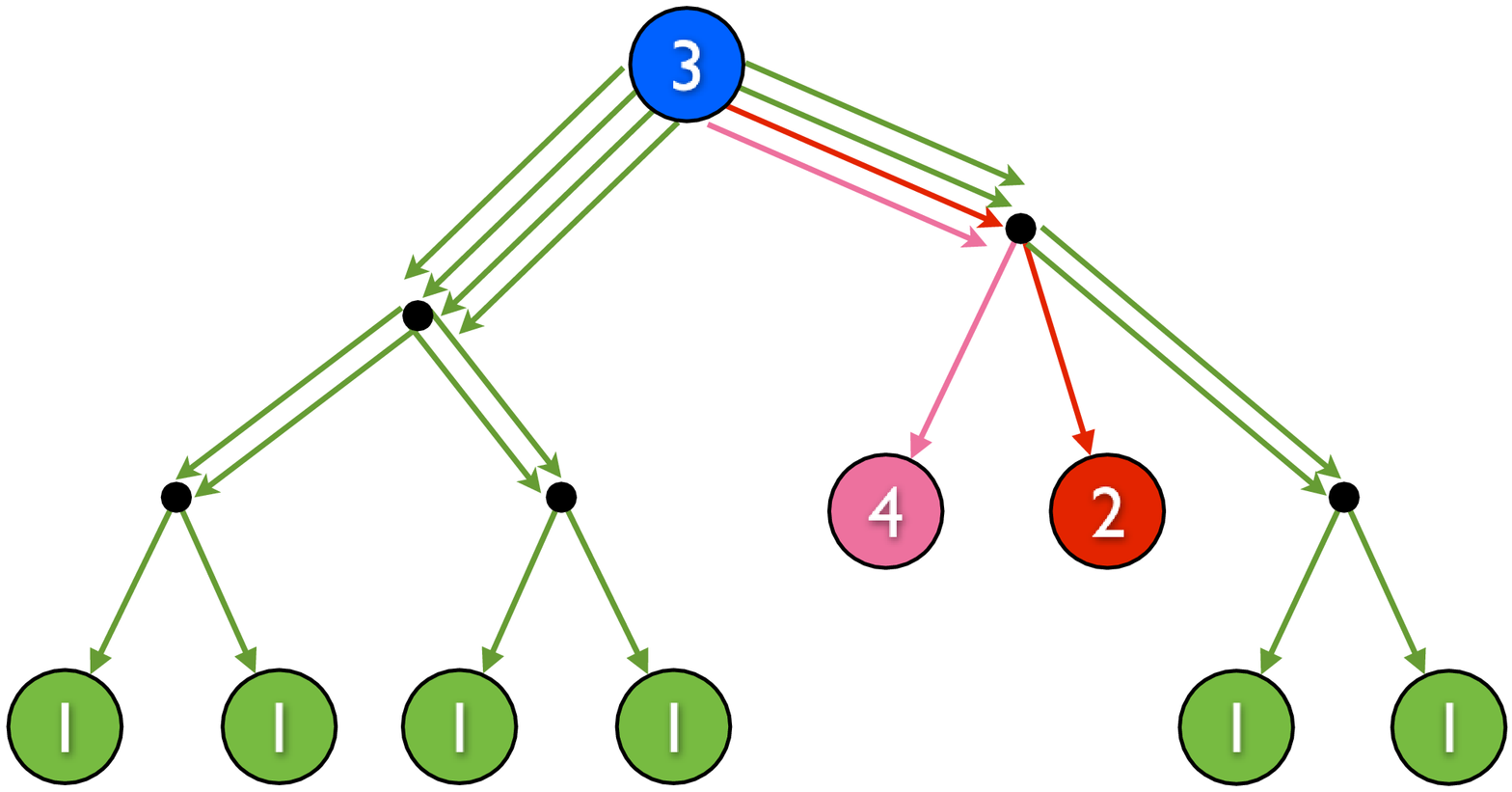}
		\vspace{-2cm}
		\caption{Decomposition of a tree into strings for the quantum double of $\mathbb{Z}_5$. The weight of each edge of the graph is represented by a colour coding.}
	\label{fig:tree2strings}
\end{figure}

We construct the error corresponding to the tree by iteratively choosing a random leaf and then applying the sequence of generalized Pauli operators which create the correct anyon at the site of the leaf and then move its conjugate anyon to the site of the root. We sequentially connect each leaf to the root. During any step, there will be at most two violations, one for the site of the conjugate anyon being moved to the site of the root and one for the anyon at the root which might not have the anyonic charge it should have in the error configuration. At the end of the procedure, the charge of the anyon at the root will be the one it should have in the reference error since the total anyonic charge of the tree is zero. Thus, trees have an energy barrier corresponding to the energy of two anyons, similar to the energy barrier of loops.

We have thus proven the following result:

\begin{theorem}[Energy barrier of Abelian Quantum Doubles]
For any $d$, the generalized energy barrier $\overline{\epsilon}$ of the quantum double of $\mathbb{Z}_d$ is at most the energy of two anyons, i.e., 
\begin{equation}
\overline{\epsilon}(H_{\mathbb{Z}_d}) \leq 2 J_{\textrm{max}}
\end{equation}
\end{theorem}

\subsection{Length of the local errors path}\label{sec:canPathLength}

We have established in the previous subsections that the optimal local errors path consists of partitioning the error into a product of errors, each of which is supported either on a loop or a union of strings. The question remains what is the length $\mu(N)$, i.e., the number of steps before the local errors path $\left\{\sigma_{\overline{\boldsymbol{\eta}}^t}\right\}$ converges to the reference error $\sigma_{\boldsymbol{\eta}}$. Formally, define $\left| \left\{\sigma_{\overline{\boldsymbol{\eta}}^t}\right\} \right|$ to be the number of operators needed to converge to a reference error $\sigma_{\boldsymbol{\eta}}$. We only consider \emph{optimal} local errors path, i.e., those which realize $\sigma_{\boldsymbol{\eta}}$ with the minimal energy barrier. We then define the optimal local length of an error to be
\begin{equation}
\mu(\boldsymbol{\eta})=\min_{\substack{ \{\sigma_{\overline{\boldsymbol{\eta}}^t}\}\to \sigma_{\boldsymbol{\eta}} \\ \max_t\overline{\epsilon}(\overline{\boldsymbol{\eta}}^t|\boldsymbol{\eta})   = \overline{\epsilon}( \boldsymbol{\eta})}} \left| \left\{\sigma_{\overline{\boldsymbol{\eta}}^t}\right\} \right|
\end{equation}
and the quantity which enters in the bound of mixing time, Eq.~\eqref{eq:mixing_time_bound} is the maximal optimal local length of errors
\begin{equation} \label{eq:def_mu}
\mu(N)=\max_{\sigma_{\boldsymbol{\eta}}} \mu(\boldsymbol{\eta})
\end{equation}

Note that this is constrained optimization:  we choose to first minimize the energy barrier and then look at the length of the local errors path realizing that minimum. This choice is dictated by the fact that the energy barrier enters the exponential in Eq.~\eqref{eq:mixing_time_bound} whereas the maximal optimal local length of errors $\mu$ is only a multiplicative constant. Nonetheless, $\mu$ is an extensive quantity since for any error, $\mu(\boldsymbol{\eta})$ is lower-bounded by twice the size of the support of the error (the factor two comes from applying the $X$ and $Z$ part of the error independently). Thus, $4N \leq \mu$. 

However, in order to minimize the energy barrier, a given qudit could be affected multiple times by single-qudit operators applied between two intermediate errors. In the language of graph, a given edge of the graph could belong to a large number of loops and trees. Indeed, one has to be careful to avoid such a phenomenon. Here we will show that the loop part of the error can be constructed with a path of length at most $4(d-1)N$ while the string part with a path of length at most $4(d-1)N$, thus the maximal optimal length is $\mu(N) \leq 8(d-1)N$.

\subsubsection{Loops} 

Given a qubit, we want to bound the number of loops which act non-trivially on that qudit. A priori, the number of loops could be very large. However, we can use a simple procedure to reduce it. The idea is to look at the weight of all edges overlapping that qudit and to identify subsets of those weights which sum to $0$ modulo $d$. In that case, we can fuse the corresponding anyons to the vacuum and get new loops which do not affect the qudit. We call this procedure \emph{merging}. An example of merging is presented on Fig.~\ref{fig:merging}. 

This procedure can be repeated on every qudit independently. The question is then to bound the number of loops at the end of merging. In Appendix~\ref{sec:zero-sum_multiset}, we investigate this question using multiset theory and find that the maximal number of loops that can remain after merging is $d-1$ (see Thm.~\ref{thm:max_sum}). Thus, after merging, any qudit belongs to at most (d-1) loops of type X and (d-1) loops of type Z. Since there are $2N$ qudits,  
\begin{equation}
\mu_{\textrm{loops}}\leq4(d-1)N.
\end{equation}

\begin{figure} 
	\centering
		\includegraphics[angle=90,origin=c,width= \columnwidth]{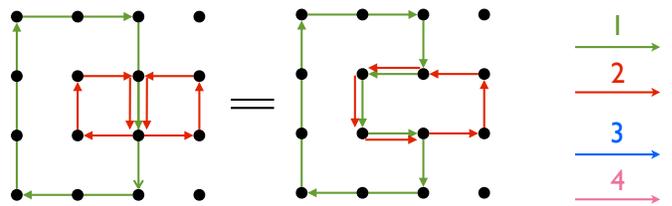} 
		\vspace{-3cm}
		
		\caption{Merging of loops at a qudit initially affected by three loops for the quantum double of $\mathbb{Z}_5$. The weight of each edge of the graph is represented by a colour coding.}
	\label{fig:merging}
\end{figure}

\subsubsection{Strings} A union of strings -- after removing the loops from the structure -- form a tree with several "leaves", the leaves corresponding to the end position of anyons.

It will be necessary to introduce a procedure to "prune the tree", i.e., decompose a tree into a superposition of subtrees without introducing new anyons. This pruning procedure was not necessary to prove that the generalized energy barrier is at most $2J_\textrm{max}$, but will prove useful to bound the length of the canonical path.

The pruning procedure identifies subtrees that can be removed from the original tree. Those subtrees should have leaves whose anyonic charge sum to zero modulo $d$ so that they can be removed without affecting the root. Before identifying those subtrees, it is convenient to first "fatten the tree" by connecting every leaf of weight $k$ to the root through a string of weight $k$. The pruning procedure then proceeds by visiting every vertex of the tree (for instance using a post-order depth first search~\footnote{This specific ordering of the vertices allows to find subtrees with small depth but is not crucial to our argument.}). At every vertex of the tree, it checks whether there exists a subset of edges with zero-sum. If so, the pruning procedure removes the subtree generated by the corresponding leaves. After visiting every vertex, the pruning stops. The tree is now decomposed into \emph{simple} trees.

Simple trees have at most $d-1$ leaves since any set of $d$ anyons contains a subset whose sum is zero modulo $d$ (see Thm.~\ref{thm:max_sum}). Their depth is thus bounded by $d-1$ Also, every vertex of a simple tree belongs to at most $d-1$ strings.  

Thus, after pruning, every qudit belongs to at most $(d-1)$ loops of type $X$ and $(d-1)$ loops of type $Z$. Since there are $2N$ qudits,  
\begin{equation}
\mu_{\textrm{strings}}\leq4(d-1)N.
\end{equation}

\subsection{The effect of defect lines}

For Abelian quantum double, it is possible to locally modify the Hamiltonian in order to introduce defect lines, such as in the work of Brown et al.~\cite{BAP14}. Defect lines are characterized by an invertible element $M\in\mathbb{Z}_d$ and an orientation. An anyon of type $k\in\mathbb{Z}_d$ crossing a defect line of type $M$ along the orientation (resp. against the orientation) will be transformed into an anyon of type $M\cdot k$ (resp. $M^{-1}\cdot k$). What do we mean by "transformed"? Consider two vertices $(v_-,v_+)$ on the lattice, one on each side of the defect line such that the orientation points from $v_-$ to $v_+$. There exists a local Pauli operator which maps a +1 eigenstate of $P^k_{v_i}$ to a +1 eigenstate of $P^{M\cdot k}_{v_i}$. In other words, an excitation will locally at energy $J_k$ become an excitation carrying energy $J_{M\cdot k}$. 

In~\cite{BAP14}, Brown et al. proposed a local 2D Hamiltonian which seemed to realize entropy protection~\cite{BAP14}. This model is the quantum double of $\mathbb{Z}_5$; and due to charge-flux duality, we're allowed to think only in terms of e.g. electric charges. Then there are $5$ different charges, grouped as: vacuum, light particle, heavy particle, heavy antiparticle, light antiparticle. Particle-antiparticle pairs have the same mass, furthermore $m_{\textrm{heavy}} > 2 m_{\textrm{light}}$ to ensure that during the thermal evolution of the system it is favorable for the heavy particles to decay into two light particles. In order to favor the occurrence of heavy particles, the authors of Ref.~\cite{BAP14} introduced defect lines of type $M=2$ to the system. When a light particle crosses such a line, it becomes a heavy one and vice versa. Thus, the excitations in the model are typically light particles which propagate freely until they eventually cross a defect line, acquire mass, and then decay into two light particles. It was observed numerically in~\cite{BAP14} that the memory time seems to behave like $t_{\textrm{mem}} \propto \exp(c\beta^2)$ over some range of parameters but seems to fail for large $\beta$.  Can our bound shed new light on this model? To that end, we now analyze the effect of those defect lines on our bounds. 

 \subsubsection{Syndromes for the Hamiltonian with defect lines}
 
One could wonder whether the definition of the energy barrier given by Eq.~\eqref{eq:energybarrier} should be changed due to the introduction of defect lines. It does not. However, the Hamiltonian changed and thus the syndromes of Pauli errors will change too. Given a Pauli error $\boldsymbol{\xi}\in\mathcal{P}_{2N}$, its syndrome with respect to the new Hamiltonian $e^{\textrm{new}}(\boldsymbol{\xi})$ is related to the syndrome $e^{\textrm{no defect lines}}(\boldsymbol{\xi})$ it had in the absence of defect lines by simply multiplying the syndrome by the \emph{defect line string}  $T_1\in\mathbb{Z}_d^{2N}$,

\begin{equation}
e^{\textrm{new}}(\boldsymbol{\xi}) = (T_1) \cdot e^{\textrm{no defect lines}}(\boldsymbol{\xi})
\end{equation}
where multiplication is understood ditwise and modulo $d$. The \emph{defect line string} $T_1\in\mathbb{Z}_d^{2N}$ is defined for every site $s$ by
\begin{equation}
\left(T_{1}\right)_{s}=\begin{cases}
M & \mbox{if \ensuremath{s} near (and on the "\ensuremath{-}" side of) a defect line } \\
1 & \mbox{otherwise}
\end{cases}
\end{equation}

Therefore, there is a consistent way to get the syndromes of the quantum double $\mathbb{Z}_d$ with defect lines, and we can use this new set of syndromes to work through the same steps in the derivation as we did for the quantum doubles without defect lines. These two derivations will essentially be identical -- except for the different definitions of the syndromes -- and we will arrive to the same formula for the energy barrier.

\subsubsection{Globally consistent labelling of anyon types in the presence of defect lines}

The only remaining question is: knowing that the definition of the energy barrier is the same with defect lines, does the evaluation of the energy barrier detailed in the previous sections go through the same way? The main issue is how to label the excitations. Indeed, due to the presence of defect lines, the \emph{local} labelling of the anyon type is not consistent \emph{globally}.

Here we explain how to recover a global labelling of anyon types, under one technical condition we call \emph{consistency} of defect lines. We define the consistency of the defect lines of a model by requiring that when we create a pair of anyons from vacuum, then take one of them around any loop anywhere on the lattice, they fuse back to vacuum with each other. Should that transparency condition be violated, the intersection of defect lines would become a sink and a source for single anyons, which we forbid. Furthermore, we do not know how the Hamiltonian of such a pathological model would be written down in a form similar to Eq.~\eqref{eq:QDoubleHamiltonian}.

Thus, we consider \textit{consistent} defect lines. Our goal is to take the globally inconsistent, \emph{local} anyon syndromes which is a record of the eigenvalues of the $A(v)$/$B(p)$ star/plaquette operators at each site, and translate them to a consistent, \emph{global} labeling of anyons. This translation is obtained through a global dictionary $T_2\in\mathbb{Z}_d$ using the formula 
\be
e^{\textrm{global}} (\boldsymbol{\xi}) = (T_2) \cdot e^{\textrm{local}} (\boldsymbol{\xi}) ,
\ee
where $\cdot$ is multiplication ditwise and modulo $d$. To define the global dictionary, the idea is to label each region enclosed by defect lines. The anyon types will be defined in one (arbitrary) reference region and all other regions will carry a label to translate the local anyon type within its region to what it would be in the reference region (global syndrome).

For instance, for $\mathbb{Z}_5$, with $M=2$ an anyon type $a$ in the reference region might become: $a$ or $2a$ or $4a$ or $8a(=3a\mod 5)$, depending in which region we observe it. We can name these regions, e.g. $L= 1$, $2$, $4$ and $3$ in the above example. Whenever we observe an anyon whose local type is $b$ in a region with a label $L$, we know that anyon would have a local type $b'=L^{-1}b$ in the reference region (or any trivial $L=1$ region). Thus, the $T_2$ dictionary is defined for every site by 
\be
(T_2)_s = L^{-1} \;\; \textrm{for $s \in$ region with label $L$} .
\ee

\subsubsection{Evaluation of the generalized energy barrier and maximum length of the optimal local errors path}

Finally, introducing defect lines doesn't change the allowed anyon fusion/decay processes either, since the fusion rules are the same as before in every region. Whenever a particle crosses a defect line it is essentially just renamed, i.e., it doesn't leave behind a charge at the defect line. Using the fact that the syndromes can be made consistent with the procedure of tracing all anyons back to the $L=1$ regions, any error can still be mapped onto a graph flow of anyons, and the plan for constructing any generalized Pauli error described in Sec.~\ref{subsec:en_barrier_construction} still works. Thus, the value of the energy barrier and the maximum length of the optimal canonical path is unchanged as well: $\bar{\epsilon}=2 J_\textrm{max}$ and $\mu(N) \leq 4N^2 + 4 (d-1) N$.

Therefore, as neither the definition of the energy barrier, nor the structure of errors, nor the optimal canonical path for a certain error, nor the length of this path is changed by defect lines, the Arrhenius law bound itself is unchanged by the defect lines.
\section{Discussion}
\label{sec:Disc} 

\subsection{Possible improvements}

We briefly review some possible improvements on our bounds, indicate possible avenues to achieve those improvements and conjecture what the optimal bounds would be.

The polynomial dependence of the Arrhenius bound on mixing time can probably be improved. Indeed, we expect that better techniques would allow to get rid of the $N$ prefactor in Eq.~\eqref{eq:mixing_time_bound}. However, the polynomial dependence of the length of the longest optimal local errors path $\mu \sim N$ is tight since one can find errors whose length are of the order of the number of qudits. Thus, we expect the mixing time to scale with system size. The extensiveness of mixing time is coherent with the intuition that some system relax locally. 

However, the quantum memory time might be much shorter than the mixing time. A dramatic example is the three-dimensional toric code whose quantum memory time is constant whereas its mixing time is exponentially long. Indeed, one of the logical operator is string-like whereas the other logical operator is supported on a 2D sheet of qudits. The expectation value of the sheet-like logical operator thermalizes in exponential time whereas the expectation value of the string-like logical operator is short-lived. We expect the quantum memory time of 2D Abelian quantum double to be a constant, independent of system size.

\subsection{Implication for entropy protection}

In \cite{BAP14}, authors investigate the quantum memory time of an Abelian quantum double with $d=5$ with defect lines. By tuning the masses of anyons, they obtain a thermal dynamic in which  the typical world lines of anyons have a fractal structure. Indeed, heavy particles, rather than  propagate,  will (with high probability) decay into two light particles propagating independently; while light particles will eventually cross a defect line, become a heavy particle (at an energy cost), which then decays into two light particles. Brown et al. numerically observe a super-exponential scaling of the memory time which they explain to be the result of this fractal structure of the world lines of excitations. It is called "entropic protection" as the world lines only have a fractal structure and thus there's a scaling energy barrier for a \emph{typical} worldine of anyons. There are, in fact, world lines taking the system to an orthogonal ground state with only a constant energy cost, however, the probability of such a world line are entropically suppressed.

Applying the result of the present paper to this model, we can see that its memory time is upper bounded by a strict Arrhenius law, with an energy barrier that has no dependence on temperature or system size, even when including the effect of permuting type defect lines. Since our bound is valid for any value of the inverse temperature $\beta$, we can see that the $\exp (c \beta^2)$ scaling observed in Brown's entropic code needs to break down for sufficiently low temperatures, as the memory time can't exceed our bound. This breakdown at low temperature was forecasted in~\cite{BAP14}. Indeed, for low temperature the thermal process resulting in fractal-like world lines of anyons is not typical anymore, since the environment can't provide the energy required for a light particle to become heavy. Rather, a light particle near a defect line won't cross, but linger there until it meets with another particle and fuse with it either to vacuum or to a heavy particle. If fused to a heavy particle, that heavy particle can then cross the defect line and lower the energy by becoming a light particle. 

The low temperature behaviour of Brown's entropic code agrees with the fact that our bound doesn't allow it to have a better than exponential memory time. The scaling observed in Ref.~\cite{BAP14} is most likely limited to the region discussed there, and needs to break down for temperatures out of that region.

One question that remains open is whether the super-exponential behaviour they observe is an artefact of their construction, e.g., of the decoder or a physical property of the model. Indeed, one could imagine that the introduction of defect lines does change the thermal behaviour of the model over some temperature region. The super-exponential scaling could then be understood as an \emph{entropic enhancement}. While this enhancement does not translate into a qualitatively different scaling at low temperature, it could introduce a multiplicative gain inside the exponential scaling.

\begin{figure} 
	\centering
		\includegraphics[angle=0,origin=c,width= \columnwidth]{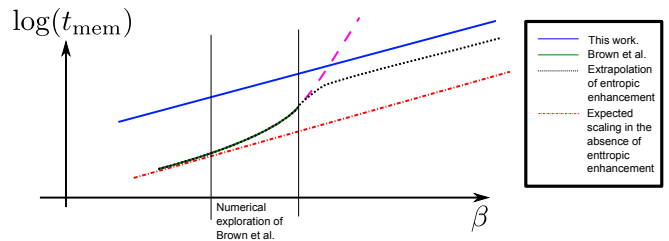} 
		
		\caption{Conceptual scaling of the logarithm of the quantum memory time $t_{\mathrm{mem}}$ as a function of the inverse temperature $\beta$. }
	\label{fig:entropy_enhancement}
\end{figure}

Our bound, however, is more general than to only exclude the possibility of entropic protection for the specific construction of Brown et al. The result presented in the present paper means that entropic protection doesn't exist for any Abelian quantum doubles (with or without permuting type defect lines); in order to have a self-correcting memory based on such models, one needs a scaling energy barrier. 
However, we should remark that a scaling energy barrier does not always ensure self-correction, as seen in the example of the welded code~\cite{Michnicki14} which is expected to have a memory time which is independent of system size~\cite{BLP+14}.



\section{Details of the derivation}
\label{sec:derivation}

The derivation of the upper bound and the generalized energy barrier for the $\mathbb{Z}_d$ generalized case follow the steps outlined in Ref.~\cite{Temme14} for the $\mathbb{Z}_2$ model.  Although the approach for the  $\mathbb{Z}_d$ - Stabilizer models is very similar to the one presented in Ref.~\cite{Temme14}, the derivation differs in several key steps from $\mathbb{Z}_2$ due to the increased complexity of the model. In this section we present the general approach of the derivation with an emphasis on the differences from the $\mathbb{Z}_2$ case.\\

To obtain the bound on the thermalization time presented in Eqn. (\ref{eq:mixing_time_bound}), we need to take two steps. First, we bound the mixing time $t_{{mix}}$ in terms of the spectral gap $\lambda$ of the Davies generator, and then we proceed to prove a lower bound on the spectral gap $\lambda$.  To obtain the bound on the mixing time in terms of the gap, we employ an upper bound to the convergence in trace norm distance derived in Ref.~\cite{TKR+10}. For any initial state $\rho_0$, that evolves according to some semi-group $\rho_t = \exp(t \cL) \rho_0$, we can bound the distance to its fixed point $\rho_G$ as $ || \rho_t - \rho_G ||_{tr} \leq \sqrt{|| \rho_G^{-1} ||} e^{-\lambda t} $. Since the Davies generator converges to the Gibbs state $\rho_G = Z^{-1}\exp(-\beta H)$ we find $|| \rho_G^{-1} || = \cO(\exp(c_0 \beta N))$, with some model specific constant $c_0$. Given a lower bound $0 < \mu \leq \lambda$ to the spectral gap, this bound would immediately imply an upper bound to the mixing time given by 

\begin{equation} \label{eq:gap_mixingtime}
t_{\textrm{mix}} \leq \mathcal{O} (\beta N \lambda^{-1}) \leq \mathcal{O} (\beta N \mu^{-1}).
\end{equation}

To arrive at a lower bound to the spectral gap of the generator $\cL$, we can make use of a variational expression for the spectral gap $\lambda$. Since the fixed point of $\cL$ is the Gibbs state, and we furthermore know that the Davies generator is Hermitian with respect to a weighted Hilbert-Schmidt inner product \cite{Davies79,Davies76}, we can express the the gap in terms of two quadratic forms. We define the Dirichlet form, $\mathcal{E}(f,f) = - \left< f, \mathcal{L}^*(f) \right>_\beta = - tr[\rho_G f^\dagger \mathcal{L}^*(f)]$ and the variance, $Var(f,f) = tr[\rho_G f^\dagger f] - |tr[\rho_G f]|^2$. With these two quadratic forms we can express the spectral gap as 

 \be
 	\lambda = \min_{f \in \cM_{d^N}} \frac{\mathcal{E}(f,f)}{Var(f,f)}.
\ee

For a simple proof of this identity the reader is referred to \cite{DS91,Fill91,TKR+10}. Hence any constant $\mu > 0$ serves as a lower bound to the spectral gap if for all $f \in \cM_{d^N}$ the Poincare inequality, $\mu Var(f,f) \leq \mathcal{E}(f,f)$, holds. Naturally, the largest possible $\mu$ coincides with $\lambda$. We will now use this inequality to derive a lower bound to the spectral gap. Note that this problem can be 
rephrased as an inequality for positive semi definite matrices. Since both $\cE(f,f)$ as well as $Var(f,f)$ are quadratic forms in $f$, we can define two matrices, $\hat{\cE}$ and $\hat{\cV}$ that correspond to the matrix 
representations of these forms. Further using the detailed balance condition we express $\cE(f,f) = \tr{f^\dagger \hat{\cE}(f)}$ and $Var(f,f) = \tr{f^\dagger \hat{\cV}(f)}$, where we have now interpreted $\cM_{d^N}$ as a Hilbert space with the canonical inner product. In this case the spectral gap can be defined as $\tau = \lambda^{-1}$, where $\tau$ is the smallest positive number so that $\tau \hat{\cE} - \hat{\cV} \geq 0$, here any upper bound to $\tau$ constitutes a lower bound to the spectral gap. We perform the following steps to find such an upper bound to $\tau$ and in turn the lower bound $\mu=\tau^{-1}$ to the gap. Due to the similarity to the $\mathbb{Z}_2$ case, the reader is referred to a \cite{Temme14} for a more detailed exposition of the steps and proofs of the lemmata we need. We discuss only the particular differences to the binary case in detail here.\\
	
\subsection{Diagonalizing the Hamiltonian, then deriving the jump operators of the Liouvillian.} \label{subsec:DiagHam_JumpOps}
	
Since the quantum double model is comprised of commuting projectors, it is straight forward to diagonalize the full Hamiltonian (\ref{eq:QDoubleHamiltonian}). We diagonalize the pure system Hamiltonian, 
by labeling the projectors for every subspace in terms of the error syndromes assigned to different error configurations introduced in Sec.~\ref{sec:quantum_doubles}. In order to be able to encode quantum information into the ground state of this Hamiltonian, a degeneracy of ground states is required. This can be achieved by defining the square lattice on a surface with non-zero genus or by special boundary conditions.

The diagonalized Hamiltonian can be written as:

\begin{equation}
H = \sum_{(\mathrm{a},\mathrm{b})} \epsilon (\mathbf{a},\mathbf{b}) \Pi (\mathbf{a},\mathbf{b}) ,
\end{equation}

with projectors and energies

\begin{eqnarray}\label{HamProj}
\Pi (\mathbf{a},\mathbf{b}) &=& \left( \prod_v P^{a(v)}_v \right) \left( \prod_p Q^{b(p)}_p \right), \\
\epsilon (\mathbf{a},\mathbf{b}) &=& \sum_v J^{a(v)}_v + \sum_p J^{b(p)}_p ,
\label{eq:HamEnergies}
\end{eqnarray}

where $P^{a(v)}_v$ and $Q^{b(p)}_p$ are the projectors onto different chargeons and fluxons at vertex $v$ and plaquette $p$ introduced before, $J^{a(v)}_v$ and $J^{b(p)}_p$ are the masses corresponding to these anyons. The eigenstates of this Hamiltonian thus correspond to different anyon configurations on the lattice, and the states can be labelled by syndromes of the form $(\mathbf{a},\mathbf{b}) = (a_1, a_2, ... a_{N}, b_1, b_2, ... b_{N}) \in Z_d^{N + N}$.

The diagonalization of $H$ implies that we can write the Gibbs state of the system as
\be
	\rho_G = \sum_{(\mathbf{a},\mathbf{b})} \rho_{\mathbf{ab}}\Pi (\mathbf{a},\mathbf{b}), \Sp \mbox{where} \Sp \rho_{\mathbf{ab}} = \frac{e^{-\beta\epsilon(\mathbf{a},\mathbf{b})}}{Z}.
\ee

Note that the Projectors $\Pi(\mathbf{a},\mathbf{b})$ are obtained as a $\bZ_d$ Fourier transform of powers of the star $A(v)$ and plaquette $B(p)$ terms as defined in Fig.~\ref{fig:operators} (a). We can write:

\begin{eqnarray}\label{eq:projector_paulibasis}
\Pi (\mathbf{a},\mathbf{b}) &=& \frac{1}{d^{2N}} \sum_{(\mathbf{x},\mathbf{y})} e^{\frac{2 \pi i}{d} \left( <\mathbf{a},\mathbf{x}> + <\mathbf{b},\mathbf{y}> \right)} \sigma_{\bar{\mathbf{x}}} \bar{\sigma}_{\bar{\mathbf{y}}},
\end{eqnarray}

where $\sigma_{\bar{\mathbf{x}}} = A^{x_1}(1) A^{x_2}(2) \cdots A^{x_N}(N)$ and $\bar{\sigma}_{\bar{\mathbf{y}}} = B^{y_1}(1) B^{y_2}(2) \cdots B^{y_N}(N)$. Observe that we have introduced new labels $\bar{\mathbf{x}}$ and $\bar{\mathbf{y}}$, which are linear functions of $\mathbf{x}$ and $\mathbf{y}$ respectively and are defined by the decomposition of $\{ A^{x_i}(i) \}$ and $\{ B^{y_i}(i) \}$ (which act on vertices and plaquettes) into generalized Pauli operators, i.e., the $\{ X_j^{\bar{x}_j} \}$'s and $\{ Z_j^{\bar{y}_j} \}$, which act on edges of the model, so that $A^{x_1}(1) A^{x_2}(2) \cdots A^{x_N}(N) = X_1^{\bar{x}_1} X_2^{\bar{x}_2} \cdots X_{2N}^{\bar{x}_{2N}} = \sigma_{\bar{\mathbf{x}}}$ and $B^{y_1}(1) B^{y_2}(2) \cdots B^{y_N}(N) = Z_1^{\bar{y}_1} Z_2^{\bar{y}_2} \cdots Z_{2N}^{\bar{y}_{2N}} = \bar{\sigma}_{\bar{\mathbf{y}}}$. \\

The jump operators of the Davies generator are generated by generalized Pauli errors acting on a single spin. The commutation relations of single generalized Pauli's with the Hamiltonian projectors (\ref{HamProj}) is given by:
\begin{equation} \label{commPauli}
Z_j^{l_j} X_j^{m_j} \Pi (\mathbf{a},\mathbf{b}) = \Pi (\mathbf{a}\oplus \mathbf{e}(l_j),\mathbf{b}\oplus \mathbf{e}(m_j)) Z_j^{l_j} X_j^{m_j} ,
\end{equation}
where $\mathbf{e}(l_j) = (0,\ldots 0, l_j, -l_j, 0 \ldots 0)$ and $\mathbf{e}(m_j) = (0,\ldots 0, m_j, -m_j, 0 \ldots 0)$ are length $N$ vectors whose only nonzero elements correspond to the vertices at the ends of edge $j=(v,v')$ and the plaquettes $(p,p')$ which contain edge $j$. Alternatively, $\mathbf{e}(l_j)$ ($\mathbf{e}(m_j)$) is the syndrome of the excited state created by applying the error $Z_j^{l_j}$ ($X_j^{m_j}$) to the vacuum state. That excited state contains two conjugate anyons of charges $\pm l_j$ (of fluxes $\pm m_j$) located on the vertices $v$ and $v'$ (plaquettes $(p,p')$ that contain the edge $j$).

Star and plaquette operators have one east edge E, one south edge S, one west edge W and one north edge N. Due to the construction of these operators: $A(v) = \{X_E \otimes X_S \otimes X^\dagger_W \otimes X^\dagger_N ; (E,S,W,N)\in star(v)\}$ for star operators  and $B(p) = \{Z_E \otimes Z^\dagger_S  \otimes Z^\dagger_W \otimes Z_N; (E,S,W,N)\subset plaquette(p)\}$ for plaquette operators, a horizontal edge $j=(v,v')$ overlaps with the $X$ operator of the star operator $A(v)$ west of edge $j$ and the $X^\dagger$ operator of the star operator $A(v')$ east of $j$. (Similarly, a horizontal edge $j$ overlaps with the $Z^\dagger$ operator of the plaquette operator north of $j$, and $Z$ operator of the plaquette operator south $j$.) Therefore, one of the nonzero elements in $\mathbf{e}(l_j)$ is $+l_j$ and the other is $-l_j$ ($+m_j$ and $-m_j$ in $\mathbf{e}(m_j)$). 

Since the generalized Pauli basis is a complete matrix basis, any one local operator at edge $j$ can be decomposed as $S_{j,(l'_j,m'_j)} = \sum_{(l_j,m_j)} [s]_{(l_j,m_j)}^{(l'_j,m'_j)} Z_j^{l_j} X_j^{m_j}$. In order to obtain the jump operators from Eqn. \eqref{eq:jump_operator_omega}, we use the commutation relations (\ref{commPauli}) and obtain

\begin{eqnarray}
S_{j,(l'_j,m'_j)} (\omega) &=& \sum_{(\mathrm{a},\mathrm{b})} \sum_{(l_j,m_j)} Z_j^{l_j} X_j^{m_j} [s]_{(l_j,m_j)}^{(l'_j,m'_j)} \Pi (\mathbf{a},\mathbf{b}) \\
&& \times \delta [\omega - \epsilon (\mathbf{a} \oplus \mathbf{e}(l_j),\mathbf{b} \oplus \mathbf{e}(m_j)) + \epsilon (\mathbf{a},\mathbf{b})]\nonumber ,
\end{eqnarray}

where $\epsilon (\mathbf{a},\mathbf{b})$ is the energy of the system before, while $\epsilon (\mathbf{a} \oplus \mathbf{e}(l_j),\mathbf{b} \oplus \mathbf{e}(m_j))$ is the energy configuration of the system after applying the thermal errors. 
Note that for ease of notation we have defined $\delta[x] = 1$, whenever $x = 0$ and $\delta[x] = 0$ otherwise.\\

We point out a couple of significant differences from the $\mathbb{Z}_2$ case. First, much like the standard Pauli operators the generalized $Z_j^{l_j}$ and $X_j^{m_j}$ generate a complete local unitary matrix basis, so that any local error can be expressed as a sum of their products. However, these matrices are not Hermitian. This means that the coefficients $[s]_{(l_j,m_j)}^{(l'_j,m'_j)}$ need to obey special constraints to ensure Hermiticity of the coupling operators $S_{j,(l'_j,m'_j)}$. As discussed before we make the particular choice that $S_{j,(l_j,m_j)} = 2^{-1/2}(Z_j^{l_j} X_j^{m_j} + h.c. )$. As we will see, this will eventually result in terms appearing in the Liouvillian that are proportional to $\1_j\Pi (\mathbf{a},\mathbf{b})$ and terms that are proportional to $Z_j^{2l_j} X_j^{2m_j} \Pi (\mathbf{a},\mathbf{b})$. We are familiar with the first kind of term from the case of $\bZ_2$ -stabilizers. However, the cross terms $Z_j^{2l_j} X_j^{2m_j} \Pi (\mathbf{a},\mathbf{b})$, do not vanish automatically unless we consider a small lift of the accidental degeneracy in the Hamiltonian spectrum. They disappear when introducing a small spatial perturbation in the masses of different particles, since the delta function $\delta[x]$ vanishes on the slighly perturbed spectrum.

With these derivations for $S_{j,(l'_j,m'_j)}(\omega)$ it is in principle possible to state the Davies generator from Eqn (\ref{eq:DaviesGen}) explicitly. Since the Hamiltonian is comprised of only local commuting terms, one can verify that after performing the sum over $\omega$ in Eqn (\ref{eq:DaviesGen}), one is left with a Lindbladian that can be written as the sum of local terms; as done, e.g. in Ref.~\cite{VRM12}. Note, however, that we're taking another approach, as the representation mentioned above is not particularly helpful for our derivation of the spectral gap as it obfuscates the underlying general algebraic structure. This structure is best understood in terms of the action of the generator $\cL$ on a suitable matrix basis.\\

\subsection{Construction of the Dirichlet matrix and $\hat{\cE}$ and the variance matrix $\hat{\cV}$}
		
In order to get a good handle on the matrix pair $(\hat{\cE},\hat{\cV})$, we need to choose a suitable matrix basis of the space $\cM_{d^N}$. It turns out that the canonical choice is also the most suitable. We define the tensor product of the generalized Pauli matrices as our basis through
\be
\bar{\sigma}_{\mathbf{k}} \sigma_{\mathbf{p}} = Z_1^{k_1} Z_2^{k_2} ... Z_{2N}^{k_{2N}} X_1^{p_1} X_2^{p_2} ... X_{2N}^{p_{2N}},	
\ee
where $Z_j^{k_j}$ and $X_j^{p_j}$ are the generalized Pauli matrices introduced in Eqn. (\ref{eq:GenPauli}). This matrix basis is orthogonal with respect to the Hilbert-Schmidt inner product, i.e., $\left(\bar{\sigma}_{\mathbf{k}} \sigma_{\mathbf{p}}|\bar{\sigma}_{\mathbf{k'}} \sigma_{\mathbf{p'}}\right) \propto \delta_{k,k'} \delta_{p,p'}$. Here we denote the vectorization of $\bar{\sigma}_{\mathbf{k}} \sigma_{\mathbf{p}}$ by $ \left|\bar{\sigma}_{\mathbf{k}} \sigma_{\mathbf{p}}\right)$. 

\subsubsection{Dirichlet matrix}

We first turn to the derivation of the Dirichlet matrix $\hat{\cE}$, since this matrix proves to be more challenging. Recall the defintion of $\cE(f,f) = - tr[\rho_G f^\dagger \mathcal{L}^*(f)]$. Due to detailed balance, it turns out to be very useful to investigate the action of the map $-\rho_G\cL(\cdot)$ on $\bar{\sigma}_{\mathbf{k}} \sigma_{\mathbf{p}}$. 

Before we state the action of the generator on a basis element, we need to introduce some shorthand notation.  For the syndrome vectors $(\mathbf{a},\mathbf{b})$ and $(\mathbf{c},\mathbf{f})$  and for the error vector $(\mathbf{k},\mathbf{p})$ , corresponding to applying the operator $\bar{\sigma}_{\mathbf{k}} \sigma_{\mathbf{p}}$ to the system, we define the functions:
\bq \nonumber
H_{(\mathbf{a},\mathbf{b}),(\mathbf{c},\mathbf{f})}^{(\mathbf{k},\mathbf{p})} &=& \gamma (\omega^{(\mathbf{k},\mathbf{p})} (\mathbf{a},\mathbf{b}) \delta [\omega^{(\mathbf{k},\mathbf{p})} (\mathbf{a},\mathbf{b}) - \omega^{(\mathbf{k},\mathbf{p})} (\mathbf{c},\mathbf{f})], \\  \nonumber
\tilde{H}_{(\mathbf{a},\mathbf{b}),(\mathbf{c},\mathbf{f})}^{(\mathbf{k},\mathbf{p})} &=& \gamma (\omega^{(\mathbf{k},\mathbf{p})} (\mathbf{a},\mathbf{b})) + \gamma (\omega^{(-\mathbf{k},-\mathbf{p})} (\mathbf{a},\mathbf{b})) \\
&&+ \gamma (\omega^{(\mathbf{k},\mathbf{p})} (\mathbf{c},\mathbf{f})) + \gamma (\omega^{(-\mathbf{k},-\mathbf{p})} (\mathbf{c},\mathbf{f})). \nonumber
\eq 
We have introduced the Bohr frequencies $\omega^{(\mathbf{k},\mathbf{p})} (\mathbf{a},\mathbf{b}) = \epsilon (\mathbf{a} \oplus \mathbf{e}(\mathbf{k}),\mathbf{b} \oplus \mathbf{e}(\mathbf{p})) - \epsilon (\mathbf{a},\mathbf{b})$. Moreover, it proves convenient to introduce an additional short hand for syndromes that are modified by an additional error as $(\mathbf{a},\mathbf{b})^{(\mathbf{k},\mathbf{p})} = (\mathbf{a} \ominus \mathbf{e}(\mathbf{k}),\mathbf{b} \ominus \mathbf{e}(\mathbf{p}))$.\\

 With this notation at hand, we can state an explicit representation of the action of the generator on the generalized Pauli basis element, so that

\begin{eqnarray}\label{eq:Liouvillian_in_Pauli_basis}
-\rho_G \mathcal{L}(\bar{\sigma}_{\mathbf{k}} \sigma_{\mathbf{p}}) &=& \sum_{(\mathbf{a},\mathbf{b})} \sum_{j,(l_j,m_j)} \rho_{\mathbf{ab}} \frac{1}{2} \Pi (\mathbf{a},\mathbf{b}) \bar{\sigma}_{\mathbf{k}} \sigma_{\mathbf{p}} \\ \nonumber
&\times&\left(H_{(\mathbf{a},\mathbf{b})^{(\mathbf{k},\mathbf{p})},(\mathbf{a},\mathbf{b})}^{(-l_j,-m_j)} e^{\frac{2\pi i}{d} (p_j l_j-k_j m_j)} \right. \\ \nonumber
&+&  H_{(\mathbf{a},\mathbf{b})^{(\mathbf{k},\mathbf{p})},(\mathbf{a},\mathbf{b})}^{(l_j,m_j)} e^{\frac{2\pi i}{d} (k_j m_j-p_j l_j)} \\
&-& \left. \frac{1}{2} \tilde{H}_{(\mathbf{a},\mathbf{b})^{(\mathbf{k},\mathbf{p})},(\mathbf{a},\mathbf{b})}^{(l_j,m_j)} \right) \nonumber .
\end{eqnarray}

Recall that we can express the projector $ \Pi (\mathbf{a},\mathbf{b}) $ in terms of a $\bZ_d$ Fourier transform over a particular subset of generalized Pauli operators. This in particular means, that we can express the action 
of $\cL$ on any generalized Pauli again as a linear combination of the same basis elements. Hence, we can read off the matrix elements in this basis directly. Since the Dirichlet matrix is essentially given by $-\rho_G \cL$, we can state it directly in the basis $ \{| \bar{\sigma}_{\mathbf{k}} \sigma_{\mathbf{p}} ) \}$ and obtain:

\begin{eqnarray}
\hat{\mathcal{E}} &=& \frac{1}{d^{2N}} \sum_{j,(l_j,m_j)} \sum_{(\mathbf{a},\mathbf{b})} \sum_{(\mathbf{k},\mathbf{p})} \sum_{(\mathbf{x},\mathbf{y})} e^{\frac{2\pi i}{d}(<\mathbf{a},\mathbf{x}>+<\mathbf{b},\mathbf{y}>)} \\ \nonumber
&\times& e^{- \frac{2\pi i}{d} < \mathbf{k},\bar{\mathbf{x}}> } \left| \bar{\mathbf{y}} \oplus \mathbf{k}, \bar{\mathbf{x}} \oplus \mathbf{p} \right) \left( \mathbf{k},\mathbf{p} \right| \\ \nonumber
 &\times& \frac{1}{2} \left(\frac{1}{2} \tilde{H}_{(\mathbf{a},\mathbf{b})^{(\mathbf{k},\mathbf{p})},(\mathbf{a},\mathbf{b})}^{(l_j,m_j)} - H_{(\mathbf{a},\mathbf{b})^{(\mathbf{k},\mathbf{p})},(\mathbf{a},\mathbf{b})}^{(l_j,m_j)} \theta_{(\mathbf{k},\mathbf{p}),(l_j,m_j)} \right.\\ \nonumber
 &-& \left. H_{(\mathbf{a},\mathbf{b})^{(\mathbf{k},\mathbf{p})},(\mathbf{a},\mathbf{b})}^{(-l_j,-m_j)} \theta_{(\mathbf{k},\mathbf{p}),(-l_j,-m_j)}\right) \rho_{\mathbf{ab}} ,
\label{eq:dirichlet_paulibasis}
\end{eqnarray}

where $\theta_{(\mathbf{k},\mathbf{p}),(l_j,m_j)} = e^{\frac{2\pi i}{d} (k_j m_j-p_j l_j)}$.\\ 

\subsubsection{Variance matrix}

If we now turn to the second matrix $\hat{\cV}$, note that the variance $Var(f,f)$ can be interpreted as the Dirichlet form of a completely depolarizing semi-group on $\cM_{d^N}$. That is we can introduce the depolarizing generator $\mathcal{D}(f) = \rho_G\tr{f} - f$. So that we can write $Var(f,f) = -\tr{\rho_Gf^\dagger \mathcal{D}(f)}$. Recall that the trace can be expressed as a twirl over generalized Pauli matrices as $\tr{f} = d^{-N} \sum_{\mathbf{kp}} (\bar{\sigma}_{\mathbf{k}} \sigma_{\mathbf{p}})^\dagger f \bar{\sigma}_{\mathbf{k}} \sigma_{\mathbf{p}}$. This identity proves quite useful in the derivation of the matrix representation of $Var(f,f)$. Following the same approach, as outlined in \cite{Temme14}, we can derive the matrix representation $\hat{\cV}$ of the variance much like the Dirichlet matrix and obtain:

\begin{eqnarray}
\hat{\cV} &=& \frac{1}{d^N} \frac{1}{d^{2N}} \sum_{(\nu,\kappa)} \sum_{(\mathbf{a},\mathbf{b})} \sum_{(\mathbf{k},\mathbf{p})} \sum_{(\mathbf{x},\mathbf{y})} e^{\frac{2\pi i}{d}(<\mathbf{a},\mathbf{x}>+<\mathbf{b},\mathbf{y}>)} \nonumber \\
&\times& \left( \rho_{\mathbf{ab}} \rho_{(\mathbf{a},\mathbf{b})^{(-\nu,-\kappa)}} - \rho_{\mathbf{ab}} \rho_{(\mathbf{a},\mathbf{b})^{(-\nu,-\kappa)}} \theta_{(\mathbf{k},\mathbf{p}),(\nu,\kappa)} \right) \nonumber \\
&\times& \left| \bar{\mathbf{y}} \oplus \mathbf{k}, \bar{\mathbf{x}} \oplus \mathbf{p} \right) \left( \mathbf{k},\mathbf{p} \right|.
\end{eqnarray}

\subsubsection{Dirichlet and variance matrices in the dual basis}

Note that the Dirichlet matrix and the Variance matrix are formally very similar. A central difference however is that the sum over $(\nu,\kappa)$  in the definition of $\hat{\cV}$  is taken over the full matrix basis $\bar{\sigma}_{\nu} \sigma_{\kappa} = Z_1^{\nu_1} Z_2^{\nu_2} \cdots Z_{2N}^{\nu_{2N}} X_1^{\kappa_1} X_2^{\kappa_2} \cdots X_{2N}^{\kappa_{2N}}$. This is considerably different from the sum over $(j,(l_j,m_j))$ in the Dirichlet matrix $\hat{\cE}$. This sum is constrained to run only over all local operators acting only on a single site. Hence $\hat{\cV}$ contains considerably more summands than $\hat{\cE}$. It is now the central challenge to show that despite this larger number of summands the span of $\hat{\cV}$  lies well within the span of $\hat{\cE}$ and the matrix can be supported with a small $\tau$. The structural similarity becomes even more evident, when we perform a convenient basis transformation. We consider the dual basis of the commuting subgroup generated by the projectors in the quantum double Hamiltonian:

\bq\nonumber
\left| (\mathbf{a},\mathbf{b})_{(\mathbf{k}_0,\mathbf{p}_0)} \right) &=& \frac{1}{d^{N}} \sum_{(\mathbf{x},\mathbf{y})} e^{\frac{2\pi i}{d} (<\mathbf{a},\mathbf{x}>+<\mathbf{b},\mathbf{y}>)} e^{-\frac{2\pi i}{d} <\mathbf{k}_0,\bar{\mathbf{x}}> } \\
&\times& \left| \bar{\mathbf{y}} \oplus \mathbf{k}_0, \bar{\mathbf{x}} \oplus \mathbf{p}_0 \right).
\eq

Note that every dual vector that starts from some particular reference state labeled by $(\mathbf{k},\mathbf{p})$, is orthogonal to all other dual states which is not contained within the left action of the commuting generator group of the Hamiltonian. That is, every dual space  spanned $\left| (\mathbf{a},\mathbf{b})_{(\mathbf{k}_0,\mathbf{p}_0)} \right)$ is orthogonal to the one spanned by $\left| (\mathbf{a},\mathbf{b})_{(\mathbf{k}'_0,\mathbf{p}'_0)} \right)$ if we can't find a $\bar{\sigma}_{\bar{\mathbf{y}}}\sigma_{\bar{\mathbf{x}}}$ so that $\bar{\sigma}_{\mathbf{k}_0'} \sigma_{\mathbf{p}_0'} \propto \bar{\sigma}_{\bar{\mathbf{y}}}\sigma_{\bar{\mathbf{x}}} \bar{\sigma}_{\mathbf{k}_0} \sigma_{\mathbf{p}_0}$. Hence we have a natural decomposition of the matrix alegbra into dual basis sets. Now, we furthermore introduce the states

\bq
\left| - _{\mathbf{ab}_{(\mathbf{kp})_0}}^{(\nu,\kappa)} \right) &=& \frac{1}{\sqrt{2}} \left( \left| (\mathbf{a},\mathbf{b})_{(\mathbf{k}_0,\mathbf{p}_0)} \right) \right.\\\nonumber
&-& \left. \theta_{(\mathbf{k}_0,\mathbf{p}_0),(-\nu,-\kappa)} \left| (\mathbf{a},\mathbf{b})_{(\mathbf{k}_0,\mathbf{p}_0)}^{(-\nu,-\kappa)} \right) \right),
\eq

where $\theta_{(\mathbf{k}_0,\mathbf{p}_0),(\nu,\kappa)} = e^{\frac{2\pi i}{d} (<\mathbf{k}_0,\mathbf{\kappa}>-<\mathbf{p}_0,\mathbf{\nu}>)}$, and recall the short hand notation $(\mathbf{a}\oplus \mathbf{e}(\nu),\mathbf{b}\oplus \mathbf{e}(\kappa)) = (\mathbf{a},\mathbf{b})^{(-\nu,-\kappa)}$. \\

With these vectors at hand, we can write the variance matrix as the direct sum over the orthogonal sets of the dual basis vectors as $\hat{\cV} = \bigoplus_{(\mathbf{kp})_0} \hat{\mathcal{V}}_{(\mathbf{kp})_0}$. Where every summand is positively weighted sum of projectors on to the $\left| - _{\mathbf{ab}_{(\mathbf{kp})_0}}^{(\nu,\kappa)} \right) $  so that  

\begin{equation}
\hat{\mathcal{V}}_{(\mathbf{kp})_0} = \frac{1}{d^N} \sum_{(\nu,\kappa)} \sum_{(\mathbf{a},\mathbf{b})} \rho_{\mathbf{ab}} \rho_{(\mathbf{a},\mathbf{b})^{(-\nu,-\kappa)}} \left| - _{\mathbf{ab}_{(\mathbf{kp})_0}}^{(\nu,\kappa)} \right) \left( - _{\mathbf{ab}_{(\mathbf{kp})_0}}^{(\nu,\kappa)} \right|.
\end{equation} 

If we transform the Dirichlet matrix into the same dual basis, we observe the same block diagonal structure over $\hat{\cE} = \bigoplus_{(\mathbf{kp})_0} \hat{\mathcal{\cE}}_{(\mathbf{kp})_0}$. One central difference to the variance is that the resulting matrices $\hat{\mathcal{\cE}}_{(\mathbf{kp})_0}$ can not always be expressed as a sum of projectors. The resulting matrices $\hat{\cE}_{(\mathbf{kp})_0}$ have more weight on the diagonal. However, we can find other matrices that lower bound them in a semi-definite sense so that  $\hat{\mathcal{E}}_{(\mathbf{kp})_0} \geq \hat{\mathcal{E}}'_{(\mathbf{kp})_0}$, where $\hat{\mathcal{E}}'_{(\mathbf{kp})_0}$ is a sum of projectors. Note that when employing this bound we only worsen our estimate of $\tau$. The lower bound $\hat{\cE}'_{(\mathbf{kp})_0}$ is now of the desired form so that we can write after a similar calculation:

\begin{eqnarray}\label{eq:dirichlet_blockdiagonal}
\hat{\cE}'_{(\mathbf{kp})_0} &=& \frac{1}{4} \sum_{j,(l_j,m_j)} \sum_{(\mathbf{a},\mathbf{b})} \rho_{\mathbf{ab}} \gamma(\omega^{(l_j,m_j)} (\mathbf{a},\mathbf{b}))\\\nonumber
&\times& \left| - _{\mathbf{ab}_{(\mathbf{kp})_0}}^{(l_j,m_j)} \right) \left( - _{\mathbf{ab}_{(\mathbf{kp})_0}}^{(l_j,m_j)} \right| .
\end{eqnarray}

\subsection{Bounds on the gap from a comparison theorem and canonical paths in the matrix algebra}

We have constructed the pair $(\hat{\cE}',\hat{\cV})$ in a suitable basis. It is our goal to find a sufficiently small constant $\tau$ so that the positive semi-definite matrix inequality $\tau \hat{\cE}' - \hat{\cV} \geq 0$ holds. Since both matrices are jointly block diagonal, we can compare them block-by-block, i.e., find $\tau_{(\mathbf{kp})_0}$ for every $(\mathbf{k}_0,\mathbf{p}_0)$ so that $\tau_{(\mathbf{kp})_0} \hat{\cE}'_{(\mathbf{kp})_0} - \hat{\cV}_{(\mathbf{kp})_0} \geq 0$ and simply choose $\tau$ to be the largest $\tau_{(\mathbf{kp})_0}$. This problem can be solved using a framework which is called support theory  \cite{BH03, CG05}. This framework was used in Ref. \cite{Temme14}, to derive an upper bound  on $\tau$ for a matrix pair, which is very similar to the one presented here. The fact that we can generalize theorem 11 in \cite{Temme14} to quantum doubles is a consequence of the following observation: \\

We have pointed out earlier that $\hat{\cE}'_{(\mathbf{kp})_0}$ and  $\hat{\cV}_{(\mathbf{kp})_0}$ are structurally very similar in that both matrices are positively weighted sums of rank one projectors $| - _{\mathbf{ab}_{(\mathbf{kp})_0}}^{(\nu,\kappa)} ) ( - _{\mathbf{ab}_{(\mathbf{kp})_0}}^{(\nu,\kappa)} |$. The difference, however, lies in the fact that for $\hat{\cE}'_{(\mathbf{kp})_0}$ we only sum over projectors that stem from single site Pauli operators labeled by $(l_j,m_j)$ for $j = 1\ldots N$, whereas in $\hat{\cV}_{(\mathbf{kp})_0}$ we sum over projectors that come from the full generalized Pauli algebra. The sum in $\hat{\cV}_{(\mathbf{kp})_0}$ is therefore significantly bigger. However, the algebra that can be constructed in both cases is the same. We can construct  every generalized Pauli $\bar{\sigma}_\nu \sigma_\kappa$ from the the product of single site generalized Pauli so that $\bar{\sigma}_{\nu} \sigma_{\kappa} = \bar{\sigma}_{l^1}\ldots\bar{\sigma}_{l^{|\nu|}} \sigma_{m^1}\ldots\sigma_{m^{|\kappa|}}$.  A local error path for a generalized Pauli $\bar{\sigma}_{\nu} \sigma_{\kappa}$ is a sequence of generalized Pauli operators starting from the identity $(\nu^0,\kappa^0) = 0 $ with $((\nu^0,\kappa^0),(\nu^1,\kappa^1),\ldots,(\nu^t,\kappa^t),\dots,(\nu,\kappa))$ and terminating in $(\nu,\kappa)$, so that any
subsequent configurations along the path $(\nu^t,\kappa^t)$ and $(\nu^{t+1},\kappa^{t+1})$ only differ by a single site generalized Pauli operator (see \ref{def:local_err_path}). With such a decomposition of generalized Pauli operators at hand, observe that any vector $| - _{\mathbf{ab}_{(\mathbf{kp})_0}}^{(\nu,\kappa)} )$ can be decomposed in single site vectors $\left| - _{\mathbf{ab}_{(\mathbf{kp})_0}}^{(l_j,m_j)} \right) $ as 
\be\label{eq:vec_decomp}
| - _{\mathbf{ab}_{(\mathbf{kp})_0}}^{(\nu,\kappa)} ) = \sum_{t=0}^{t_{max} -1}  \theta_{(\mathbf{k}_0,\mathbf{p}_0),(-\nu^{t+1},-\kappa^{t+1})} | - _{(\mathbf{a}, \mathbf{b})^{(-\nu^t, -\kappa^t)}_{(\mathbf{kp})_0}}^{(l^{t+1},m^{t+1})} ),
\ee
where the labels $(\nu^{t+1},\kappa^{t+1})$ and $(\nu^{t},\kappa^{t})$ differ by the single site labels $(l^{t+1},m^{t+1})$. This decomposition lies at the center of the comparison theorem 11 in Ref.~ \cite{Temme14}. \\

In order to state the result of this comparison theorem, we need to define quantum canonical paths. Observe that the decomposition in Eqn. (\ref{eq:vec_decomp}) not only depends on the partially constructed Pauli, but also on the syndromes, or the excitations the path starts from initially. To obtain a valid decomposition we need to keep track of the excitations as well. We therefore define a {\it quantum canonical path} to consist of a series of labels 
\begin{eqnarray}
\hat{\eta}_{(\mathbf{a},\mathbf{b})} = [\{(\mathbf{a},\mathbf{b}),0\}, \{(\mathbf{a},\mathbf{b})^{(-\nu^1,-\kappa^1)},(\nu^1,\kappa^1)\}, \\
\nonumber
\ldots \{(\mathbf{a},\mathbf{b})^{-\boldsymbol{\eta}},\boldsymbol{\eta}\}],
\end{eqnarray}
where the first of the labels  $(\mathbf{a},\mathbf{b})$ correspond to syndromes (excitations) and the second label $\boldsymbol{\eta}=(\nu,\kappa)$ corresponds to a partially constructed generalized Pauli operator. While edges correspond to single qudit errors present in the Dirichlet form, the whole path corresponds to a general error appearing in the variance. That is, at each link $\hat{\xi} = [ \{(\mathbf{a},\mathbf{b})^{(-\nu^k,-\kappa^k)},(\nu^k,\kappa^k)\}, \{(\mathbf{a},\mathbf{b})^{-(\nu^k \oplus l^{k+1},\kappa^k\oplus m^{k+1})},(\nu^k \oplus l^{k+1},\kappa^k\oplus m^{k+1})\} ]$ two subsequent Pauli operators differ only by a single site operator. Assume now we choose for every syndrome (set of excitations) $(\mathbf{a},\mathbf{b})$ and every generalized Pauli $\boldsymbol{\eta}$ a canonical path $\hat{\eta}_{(\mathbf{a},\mathbf{b})}$. Even though the quality of the bound strongly depends on the particular choice of this decomposition, we acquire valid bounds for any choice of $\hat{\eta}_{(\mathbf{a},\mathbf{b})}$.
We have now all components in place to follow the proof of theorem 11 in Ref. \cite{Temme14} to obtain the upper bound on $\tau$ as:

\be\label{canBound}
\tau \leq  \max_{\hat{\xi}} \frac{4 \mu(N)}{d^N \rho_{\mathbf{ab}} \gamma(\omega^{(l,m)} (\mathbf{a},\mathbf{b})) } \sum_{\hat{\eta}_{(\mathbf{a}',\mathbf{b}')} \ni  \hat{\xi} } \rho_{\mathbf{a}'\mathbf{b}'} \rho_{{(\mathbf{a}',\mathbf{b}')}^{-\boldsymbol{\eta}}}.
\ee

The maximum is taken over all possible edges $\hat{\xi}$. The sum is taken over all canonical paths $\hat{\eta}_{(\mathbf{a}',\mathbf{b}')}$ that traverse the edge $\hat{\xi}$. That is we sum over syndromes $(\mathbf{a}',\mathbf{b}')$ and errors $\eta$ that contain the edge $\hat{\xi}$ in their canonical path $\hat{\eta}_{(\mathbf{a}',\mathbf{b}')}$. Moreover observe that the bound also depends on the length $\mu(N)$ of the largest canonical path which has been analyzed in section \ref{sec:canPathLength}. We pause to observe that this bound is very similar to the canonical paths bound for graph laplacians as given derived in \cite{SJ89,DS91,Fill91}. However this bound has been obtained for a full quantum mechanical semi-group and the paths are constructed from the multiplication rules of a matrix algebra.\\

\subsection{Evaluation of the bound and the generalized energy barrier}
	
The similarity of this bound to the classical canonical paths gives rise to a convenient way of evaluating the upper bound in Eqn. (\ref{canBound}). We use the approach introduced in Ref.~\cite{SJ89}. To evaluate the bound we need to introduce a map $\Phi_\xi $ that maps any $\hat{\eta}_{(\mathbf{a},\mathbf{b})}$ that makes use of the link $\hat{\xi} = [\{(\mathbf{a},\mathbf{b})^{-\boldsymbol{\xi}}, \boldsymbol{\xi}\},\{(\mathbf{a},\mathbf{b})^{-(\boldsymbol{\xi} \oplus (l,m))}, {\boldsymbol{\xi} \oplus (l,m)}\}]$ to a corresponding Pauli. We define this map through 
\be
\Phi_\xi (\hat{\eta}_{(\mathbf{a},\mathbf{b})}) = \boldsymbol{\eta} \ominus \boldsymbol{\xi}.
\ee
Note that this map from the set of paths into the set of generalized Paulis is injective. This means that given the edge $\hat{\xi}$ and the image $\Phi_\xi (\hat{\eta}_{(\mathbf{a},\mathbf{b})})$ we can trivially recover the path through $\boldsymbol{\eta} = \Phi_\xi (\hat{\eta}_{(\mathbf{a},\mathbf{b})}) \oplus \boldsymbol{\xi}$ and the error syndrome $(\mathbf{a},\mathbf{b})$, since this pair uniquely identifies a path. We can now apply the argument of Ref.~\cite{SJ89}, and try to find a constant $\overline{\epsilon}$ so that for all edges $\hat{\xi}$ and all paths $\hat{\eta}_{(\mathbf{a},\mathbf{b})}$ the following inequality holds

\begin{eqnarray}
\frac{\rho_{\mathbf{a}'\mathbf{b}'} \rho_{{(\mathbf{a}',\mathbf{b}')}^{-\boldsymbol{\eta}}}} {\gamma(\omega^{(l,m)} (\mathbf{a},\mathbf{b})) \rho_{\mathbf{ab}}}
\leq \frac{\exp \left( 2 \beta \bar{\epsilon} \right)}{\gamma^* } \rho_{{(\mathbf{a}',\mathbf{b}')}^{\Phi_\xi (\hat{\eta}_{(\mathbf{a}',\mathbf{b}')})}}.
\end{eqnarray}

We have denoted $\gamma^*$ to correspond to the smallest value of $\gamma(\omega^{(l,m)} (\mathbf{a},\mathbf{b}))$ over all permitted transitions. Such a constant can be found and we define 
\be
\overline{\epsilon} = \max_{\hat{\eta}_{(\mathbf{a}',\mathbf{b}')}} \overline{\epsilon} (\hat{\eta}_{(\mathbf{a}',\mathbf{b}')}),
\ee 

where for every error we have that: 

\begin{eqnarray}\nonumber
\overline{\epsilon} (\hat{\eta}_{(\mathbf{a}',\mathbf{b}')}) &=& \max_{\hat{\xi} \in \hat{\eta}_{(\mathbf{a}',\mathbf{b}')}} \left( \epsilon({(\mathbf{a}',\mathbf{b}')}^{-\boldsymbol{\xi}}) + \epsilon({(\mathbf{a}',\mathbf{b}')}^{\boldsymbol{\xi} \ominus \boldsymbol{\eta}})\right. \\
&&- \left.\epsilon((\mathbf{a}',\mathbf{b}')) - \epsilon({(\mathbf{a}',\mathbf{b}')}^{-\boldsymbol{\eta}}) \right).
\label{eq:enbarrier_masses}
\end{eqnarray}

In fact, the constant $\overline{\epsilon}$ was chosen directly so that the inequality above is satisfied for all paths and all edges. This inequlity can be now used to estimate an upper bound to  Eqn.~(\ref{canBound}),  and we find that for all $\hat{\xi}$  

\begin{eqnarray}
\tau &\leq& \frac{4 \mu(N)}{\gamma^*} \exp \left( 2 \beta \bar{\epsilon} \right) \max_{\hat{\xi}} \sum_{\hat{\eta}_{(\mathbf{a}'\mathbf{b}')} \ni \hat{\xi}} \frac{\rho_{{(\mathbf{a}',\mathbf{b}')}^{-\Phi_\xi (\hat{\eta}_{(\mathbf{a}'\mathbf{b}')})}}}{d^N} .
\end{eqnarray}

Since, the map  $\Phi_\xi$ is injective for every edge $\hat{\xi}$, we can only reach a subset of all generalized Pauli operators. Hence, we may bound the sum over this subset by summing over every generalized Pauli, so  that we may bound 
\begin{equation}
\sum_{\hat{\eta}_{(\mathbf{a}'\mathbf{b}')} \ni \hat{\xi}} \frac{1}{d^N} \rho_{{(\mathbf{a}',\mathbf{b}')}^{-\Phi_\xi (\hat{\eta}_{(\mathbf{a}'\mathbf{b}')})}} \leq \sum_{\textrm{every}  \eta} \frac{1}{d^N} \rho_{{(\mathbf{a}',\mathbf{b}')}^{-\boldsymbol{\eta}}} = 1.
\end{equation}

The last equality follows from the representation of the trace, as presented in step 2. of this section and an argument taken from \cite{Temme14}, section IV. Since this bound is independent of the choice of edge $\hat{\xi}$ we obtain the convenient bound on $\tau$ that only depends on the generalized energy barrier and the length of the longest path, 

\begin{equation}
\tau \leq \frac{4 \mu(N)}{\gamma^*} \exp \left( 2 \beta \overline{\epsilon} \right) .
\end{equation}
	
On first sight, this bound looks identical to the bound that was obtained for $\bZ_2$ - stabilizers. However, the generalized energy barrier is quite different. It does reduce to the one defined in \cite{Temme14}, when we set $d=2$, but the advance is that it now holds for all possible Abelian quantum double models. If we subsitute the energies $\epsilon({(\mathbf{a},\mathbf{b})})$ as defined in Eq.~(\ref{eq:HamEnergies}), we obtain  

\begin{eqnarray}
&&\overline{\epsilon} (\hat{\eta}_{(\mathbf{a}',\mathbf{b}')}) =  \\
&&\max_{\hat{\xi} \in \hat{\eta}_{(\mathbf{a}',\mathbf{b}')}} \left(\sum_v \left( J_v^{{\mathbf{a}'}^{-\xi_{el}}} + J_v^{{\mathbf{a}'}^{\xi_{el} \ominus \eta_{el}}} - J_v^{\mathbf{a}'} - J_v^{{\mathbf{a}'}^{-\eta_{el}}} \right) \right.  \nonumber \\
&&\left. \sum_p \left( J_p^{{\mathbf{b}'}^{-\xi_f}} + J_p^{{\mathbf{b}'}^{\xi_f \ominus \eta_f}} - J_p^{\mathbf{b}'} - J_p^{{\mathbf{b}'}^{-\eta_f}} \right) \right) \nonumber.
\end{eqnarray}

We can write $(\mathbf{a}',\mathbf{b}') = (\mathbf{a},\mathbf{b}) \ominus e(\boldsymbol{\xi})$ and similarly $\mathbf{a}' = \mathbf{a} \ominus e(\xi_{el})$, $\mathbf{b}' = \mathbf{b} \ominus e(\xi_f)$, where $\xi_{el}$, $\eta_{el}$, $\xi_f$, $\eta_f$ are the electric/magnetic part of the errors: $\boldsymbol{\xi} = (\xi_{el},\xi_f)$ and $\boldsymbol{\eta}=(\eta_{el},\eta_f)$. $\xi_{el}$ ($\eta_{el}$) is understood as $\bar{\sigma}_{\xi_{el}}$ ($\bar{\sigma}_{\eta_{el}}$) error applied to a state, while $\xi_f$ ($\eta_f$) stands for applying the $\sigma_{\xi_f}$ ($\sigma_{\eta_f}$) error. This is a direct consequence of the charge flux duality. We observe, that the charge / flux contributions behave formally identical and the contribution to the energy barrier can be seen as the sum of both these contributions, i.e. $\overline{\epsilon} (\hat{\eta}_{(\mathbf{a}',\mathbf{b}')}) = \overline{\epsilon}^{\mathbf{a}} (\hat{\eta}_{(\mathbf{a}',\mathbf{b}')}) + \overline{\epsilon}^{\mathbf{b}}(\hat{\eta}_{(\mathbf{a}',\mathbf{b}')})$. It therefore suffices to only discuss one sector, i.e. either the chargeon of the fluxon part of the model from now on and we write:

\begin{eqnarray}\label{eq:charge}
\overline{\epsilon}^{\mathbf{a}} (\hat{\eta}_{(\mathbf{a}',\mathbf{b}')}) = \max_{\hat{\xi} \in \hat{\eta}_{(\mathbf{a}',\mathbf{b}')}} \sum_v \left( J_v^{{\mathbf{a}'}^{-\xi_{el}}} + J_v^{{\mathbf{a}'}^{\xi_{el} \ominus \eta_{el}}} \right.\\\nonumber
\left.- J_v^{\mathbf{a}'} - J_v^{{\mathbf{a}'}^{-\eta_{el}}} \right) .
\end{eqnarray}

We can evaluate this barrier  as follows:  We can write $J_v^\alpha = \sum_z J_v^z \delta_{z,\alpha}$ for convenience. To this end we can express Eq.~(\ref{eq:charge}) as 

\begin{eqnarray}\label{eq:energy_barrier_z_dependent}
\overline{\epsilon}^{\mathbf{a}} (\hat{\eta}_{\mathbf{ab}}) = \max_{\xi} \sum_{v \; : \; z=0}^{d-1} J_v^z \left( \delta_{z,a_v\oplus e_v(\xi)}  -\delta_{z,a_v} \right. \\\nonumber \left.  + \delta_{z,a_v\oplus e_v(\eta) \ominus e_v(\xi)} - \delta_{z,a_v\oplus e_v(\eta)} \right) .
\end{eqnarray}
Even though this expression (and the following expressions) directly only incorporates the electric sector, i.e. $\xi = \xi_{el}$ and $\eta= \eta_{el}$ above, it is still the complete energy barrier. Due to the charge-flux duality, we can construct the electric and magnetic errors one after the other, therefore at any time we need only to look at one of the sectors.

In order to evaluate this barrier, we need to consider several different scenarios in order to express this equation more conveniently. The different cases correspond to different values of $e_v(\xi)$ and $e_v(\eta)$ and are summarized in the table. The value inside the parentheses of course depends on the relative value of $a_v$ and $z$. Our goal is to get a bound that holds for all possible $a_v$ starting configurations, thus we maximize expression (\ref{eq:energy_barrier_z_dependent}) as a function of $a_v$. In order to achieve that maximum, we have chosen the relative value of $a_v$ and $z$ such that it gives the highest possible value in each case.

\begin{table}[htbp]
\begin{center}
\begin{tabular}{ c | c | c | c }
$e_v(\xi)$ & $e_v(\eta)$ & Sum of $\delta$'s & $z$ \\
\hline
0 & 0 & 0 & any \\
$p$ & $p$ & 0 & any \\
0 & $p$ & 0 & any \\
$p$ & 0 & 1 & $a_v \oplus e_v(\xi)$ or $a_v \ominus e_v(\xi)$ \\
$p$ & $q$ & 1 & $a_v \oplus e_v(\xi)$ or $a_v \ominus e_v(\xi) \oplus e_v(\eta)$
\end{tabular}
\end{center}
\caption{Here $p\neq q$, $p\neq 0$, $q\neq 0$.}
\end{table}
The case-to-case scenario shown in the table can be summarized in one simple formula, and we can rewrite the sum of the Kronecker symbols as the more convenient expression:
\begin{equation}
\overline{\delta_{e_v(\xi),0}} \cdot \overline{\delta_{e_v(\xi),e_v(\eta)}} \left( \delta_{z,a_v \oplus e_v(\xi)} + \delta_{z,a_v \ominus e_v(\xi) \oplus e_v(\eta)} \right) ,
\end{equation}
where $\overline{\delta_{x,y}} = 1-\delta_{x,y} = \{ 0 \textrm{\; if $x=y$}; 1 \textrm{\; if $x\neq y$} \}$.

Using this formula considering that for a canonical path $\hat{\eta}_a$ we can consider the edge $\xi = {\eta}^t$ as the partially constructed Pauli operator at some step $t$. To this 
end we need to consider for every $\hat{\eta}_a$ the largest contribution along the path and therefore have to maximise this value for the largest value $t$. With this substitution, the contributions $\overline{\epsilon}^{\mathbf{a}} (\hat{\eta})$ looks like:
\begin{eqnarray} \label{eq:mstar}
\overline{\epsilon}(\hat{\eta}) = \max_{t} \left( \max_{z,v'} J_{v'}^z \right) \sum_{v} \overline{\delta_{e_v(\bar{\eta}^t),0}} \cdot \overline{\delta_{e_v(\bar{\eta}^t),e_v(\eta)}} ,
\end{eqnarray}
i.e. the energy barrier is the maximum energy cost the environment has to provide to the system during any canonical path which constructs the error configuration $\eta$. However, since this energy barrier upper bounds the mixing time, in order to get a better upper bound we may choose the canonical path wisely, i.e. so it gives a smaller energy barrier.  The reader will observe, that this energy barrier corresponds exactly to the one that was analyzed in detail in section \ref{sec:energy_barrier} of this paper. We point to a notable difference to the analysis of only $\mathbb{Z}_2$ models. In these models only those sites contribut where the charge disappeared to vacuum, while for the $d>2$ general cases a site contributes even if the anyon doesn't completely disappear at the end of the canonical path, but it changes to an anyon characterized by a different syndrome, be it the vacuum or anything else.

This energy barrier gives a valid bound on the mixing time for all choices of canonical paths we can make, but the quality of the bound depends on the choice of the canonical path. To this end, when evaluating the bound, we follow the decomposition into single site Paulis as it was amply discussed in section \ref{subsec:en_barrier_construction}.

\section{Conclusions}
\label{sec:Concl} 

We have established a strict Arrhenius law upper bound for the memory time of all quantum doubles based on an Abelian group and gave a mathematically proper definition for the energy barrier. We have also seen that the energy barrier is a constant for these models. We may apply our results to the model introduced in Ref.~\cite{BAP14} to evaluate whether entropic protection is possible for such models.

Even though our bound on the mixing time is quite general in the sense that it applies to any Abelian quantum double, there are a variety of models to which our analysis doesn't apply. For these models the possibility of entropy protection is not yet excluded. One of the possible directions one can go is to invent different kind of defects, other than the type referred to here as "consistent" (defects that allow a consistent labeling of anyons based on the region they stay at) and "permuting"-type (which only apply a permutation to any particle crossing a defect line). 
We investigated Hamiltonians with consistent defect lines. However, interesting constructions use non-consistent defect lines to introduce topological defects, such as the construction of Bombin for the toric code with twists \cite{Bombin10}.  Another possibility we can't exclude is to assume a different thermal environment, and thus use a different noise model for this analyis. This might result in the simple permuting-type defect lines introducing entropic protection to Abelian systems. Our analysis only applies to Abelian quantum doubles, therefore, entropic protection of quantum doubles based on non-Abelian groups (or of models that are not quantum doubles of any group) is not ruled out, especially since one can think of a variety of defect lines which can arise in such models~\cite{BSW11}.


\section{Acknowledgments}

We thank Benjamin Brown for helpful discussions and for valuable comments on an early version of the draft. We acknowledge funding provided by the Institute for Quantum Information and Matter, an NSF Physics Frontiers Center (NFS Grants PHY-1125565 and PHY-0803371) with support of the Gordon and Betty Moore Foundation (GBMF-12500028). OLC is partially supported by Fonds de Recherche Quebec-Nature et Technologies and the Natural Sciences and Engineering Research Council of Canada (NSERC).


\appendix 

\section{Maximum cardinality and sum of multisets without a zero-sum subset} \label{sec:zero-sum_multiset}

In this section, we derive elementary facts about multisets of $\mathbb{Z}_d$ without any subset sum equal to zero. This is motivated by the problem of fusing anyons in a $\mathbb{Z}_d$ quantum double. Indeed, consider a set of anyons labelled by their anyonic charge. Since there can be several anyons of the same anyonic charge, we are interested in set-like mathematical objects where multiplicity is explicit. For instance, for a set of two anyons of type 1 and one anyon of type 3, we would like to write $\{1,1,3\}$. 

The formal mathematical object for this intuitive notion are called multisets. In our simple case, they are simply a multiplicity function 
\begin{equation}
f:\mathbb{Z}_d\to \mathbb{N}
\end{equation}
 where $f(k)$ is the multiplicity of $k$ in the multiset.  For instance, for $d=5$, the multiset $\{1,1,3\}$ is equivalent to $\{(1,2),(2,0),(3,1),(4,0),(5,0)\}$ which is the graph of the multiplicity function. 
 
 The cardinality of a multiset $|f|$ is the total number of elements, taking multiplicity into account, i.e., \begin{equation}
 |f|=\sum_{k\in\mathbb{Z}_d}f(k)\in\mathbb{N}.
 \end{equation} 
 The sum of a multiset $s(f)$ is the sum of all its elements, taking multiplicity into account, i.e., \begin{equation}
 s(f)=\sum_{k\in\mathbb{Z}_d}kf(k)\in\mathbb{Z}_d. 
 \end{equation}
 Note that $|f|$ is an integer whereas $s(f)$ is defined modulo $d$. Physically, the sum $s(f)$ is the anyon type resulting from fusing all anyons in the multiset. Moreover, one defines the sum $f+g$ of two multisets $f$ and $g$
 \begin{equation}
 (f+g)(k)=f(x)+g(x).
 \end{equation} 
It corresponds to the intuitive idea of adding the elements of $f$ and $g$. For instance, $\{1,1,3\}+\{1,2,4\}=\{1,1,3,1,2,4\}=\{1,1,1,2,3,4\}$. Finally, one defines a (multi)subset by 
\begin{equation}
f\subseteq g \Leftrightarrow  \forall k \in \mathbb{Z}_d \quad f(k)\leq g(k).
\end{equation}

Given a multiset of anyons, we would like to know whether it contains subsets which fuse to the vacuum. Mathematically, given a multiset $f$, we are interested in the sum of its subsets  $s(f')$ where $f'\subseteq f$. More precisely, we want to know if a subset sums to zero modulo $d$. We  define the spectrum of a multiset 
\begin{equation}
\textrm{sp}(f)=\bigcup_{f'\subseteq f, f'\neq \emptyset} s(f').
\end{equation}
and say that a multiset is \emph{zero-sum free} if $0$ is not in its spectrum, i.e. no non-empty subset sums to 0 modulo $d$. We aim to determine the largest possible cardinality and sum of a zero-sum free multiset. This is related to a well-studied problem in complexity theory and cryptography, called the subset sum problem \cite{CLR01}. 

First, we want to understand what happens to the spectrum when we add a singleton to the multiset, i.e., we consider the operation $f\to f+\{x\}$. The non-empty subsets of $f+\{x\}$ are $\{x\}$, the subsets of $f$ and the subsets of $f$ to which we add the element $x$. Thus, we have
\begin{equation} \label{eq:spectrum_change}
\textrm{sp}(f+\{x\})=\{x\}\cup\textrm{sp}(f)\cup\left(\textrm{sp}(f)+x\right).
\end{equation}

We can then prove the following lemma
\begin{lemma} \label{lemma:spectrum_increase}
The spectrum of a zero-sum free multiset strictly increases when adding any singleton, i.e.,
\begin{equation} \label{eq:spectrum_increase}
\textrm{f is zero-sum free} \Rightarrow \textrm{sp}(f) \subsetneq \textrm{sp}(f+\{x\})
\end{equation}
\end{lemma}

\begin{proof}
We prove the contrapositive of Eq.~\eqref{eq:spectrum_increase}, i.e. we consider a multiset $f$ for which $\textrm{sp}(f) = \textrm{sp}(f+\{x\})$ and we will prove that it contains a zero-sum subset. Using Eq.~\eqref{eq:spectrum_change}, the equality of spectra implies that i) x is an element of $\textrm{sp}(f)$ and ii) $\textrm{sp}(f) \subseteq \textrm{sp}(f)+\{x\}$. Since the addition by $x$ only shifts the spectrum the two sets have the same cardinality and $\textrm{sp}(f)= \textrm{sp}(f)+\{x\}$. In particular, since $x$ is an element of $\textrm{sp}(f)$, there exists a subset  $f^\star \subseteq f$ for which the following equality holds modulo $d$: $x=\textrm{sp}(f^\star)+x$. Thus, the sum of $f^\star$ is zero.
\end{proof}

Using lemma~\ref{lemma:spectrum_increase}, we can deduce the maximal cardinality of a zero-sum free multiset. Indeed, consider a zero-sum multiset by sequentially adding its elements to the empty set. The spectrum will increase at each addition of a singleton by at least 1. However, a spectrum is contained in $\mathbb{Z}_d$ and thus has at most $d-1$ elements. Hence, we have proven that
\begin{theorem}
\begin{equation}
\textrm{f is zero-sum free} \Rightarrow |f| \leq (d-1)
\end{equation}
\end{theorem}
In fact, we can saturate the bound. Any multiset of an integer $k$ coprime with $d$ and multiplicity $d-1$ is zero-sum free. In particular, the multiset containing $d-1$ with multiplicity $d-1$ maximizes the sum of any zero-sum free multiset.

\begin{theorem}\label{thm:max_sum}
\begin{flalign}
\max_{\textrm{zero-sum free }f} |f|=d-1 \\
\max_{\textrm{zero-sum free }f} s(f)=(d-1)^2
\end{flalign}
\end{theorem}

\end{document}